\documentclass[draftcls,onecolumn]{IEEEtran}

\usepackage{amsmath, amsthm, amsfonts, amssymb, amsbsy}
\usepackage{setspace}
\usepackage{graphicx}
\usepackage{pstricks,arydshln}
\usepackage{algorithm}
\usepackage{multirow}

\def\beginofproof{\noindent{\it Proof: }}
\def\endofproof{\hfill\rule{6pt}{6pt}}

\newtheorem{theorem}{Theorem}

\newtheorem{definition}{Definition}

\newtheorem{corollary}{Corollary}

\newtheorem{example}{Example}

\newcommand{\twotriangle}{\hfill $\bigtriangleup \bigtriangleup$  }
\newcommand{\eax}{\twotriangle  \end{example}}
\newcommand\bim{\begin{itemize}}
\newcommand\eim{\end{itemize}}

\newcommand\bH{{\bf H}}
\newcommand\bM{{\bf M}}

\newcommand\bO{{\bf O}}
\newcommand\bA{{\bf A}}
\newcommand\bC{{\bf C}}

\newcommand\bI{{\bf I}}

\newcommand\bB{{\bf B}}
\newcommand\bP{{\bf P}}

\newcommand\bV{{\bf V}}
\newcommand\bW{{\bf W}}
\newcommand\ba{{\bf a}}
\newcommand\bb{{\bf b}}

\newcommand\bh{{\bf h}}

\title{
Low-Density Arrays of Circulant Matrices: Rank and Row-Redundancy Analysis, and Quasi-Cyclic LDPC Codes}

\begin{document}

\author{
Qin Huang$^1$ and  Keke Liu$^2$ and Zulin Wang$^1$\\
$^1$School of Electronic and Information Engineering\\
Beihang University\\
Beijing 100083, China\\
(email:qhuang.smash@gmail.com; wzulin@vip.sina.com)\\
$^2$Department of Electrical and Computer Engineering\\
University of California, Davis\\
Davis, CA 95616\\
(email:kkeliu@ucdavis.edu)
}

\maketitle

\begin{abstract}
This paper is concerned with general analysis on the rank and row-redundancy of an array of circulants whose null space defines a QC-LDPC code. Based on the Fourier transform and the properties of conjugacy classes and Hadamard products of matrices, we derive tight upper bounds on rank and row-redundancy for general array of circulants, which make it possible to consider row-redundancy in constructions of QC-LDPC codes to achieve better performance. We further investigate the rank of two types of construction of QC-LDPC codes: constructions based on Vandermonde Matrices and Latin Squares and give combinatorial expression of the exact rank in some specific cases, which demonstrates the tightness of the bound we derive. Moreover, several types of new construction of QC-LDPC codes with large row-redundancy are presented and analyzed.
\end{abstract}

\section{Introduction}

   Quasi-cyclic (QC) codes have been a challenging and ongoing research subject in algebraic coding theory since their introduction in late 1960's [1]. These codes asymptotically achieve the Varshamov-Gilbert bound [2].  Recent research of these codes has been focused on a subclass of these codes, known as QC low-density parity-check (LDPC) codes.

   LDPC codes were first discovered by Gallager in 1962 [3] and then rediscovered in late 1990's [4], [5].  Ever since their rediscovery, a great deal of research effort has been expended in design, construction, structural and performance analysis, encoding, decoding, generalizations, and applications of LDPC codes. They have been shown to achieve the Shannon capacities for a wide range of channels with iterative decoding based on belief propagation.

   Major methods for constructing LDPC codes can be divided into two general categories, graph-theoretic based and algebraic methods. Each type of constructions has its advantages and disadvantages in terms of overall performance, encoding and decoding implementations.  In general, algebraically constructed LDPC codes have lower error-floors and their decoding using iterative message-passing algorithms, such as the sum-product algorithm (SPA) and the min-sum algorithm (MSA) converge faster than the LDPC codes of the length and rates constructed using the graph-theoretic-based methods.  Furthermore, it is much easier to construct algebraic LDPC codes with large minimum distances.  Algebraic constructions of LDPC codes are mainly based on finite fields, finite geometries, and combinatorial designs.  These constructions result in mostly QC-LDPC codes.

   QC-LDPC codes have advantages over other types of LDPC codes in hardware implementation of encoding and decoding. Encoding of a QC-LDPC code can be efficiently implemented using simple shift registers with complexity linearly proportional to its number of parity-check symbols (or its length) [6].  In hardware implementation of its decoder, the quasi-cyclic structure of the code (or circular structure of its parity-check matrix) simplifies the wire routing for message passing [7] and allows partially parallel decoding [8] which offers a tradeoff between decoding complexity and decoding speed. Furthermore, well designed or constructed QC-LDPC code can perform as well as any other types of LDPC codes.  Most of LDPC codes adopted as standard codes for various next generations of communication systems are QC-LDPC codes.

   A $q$-ary QC-LDPC code is given by the null space of an array $\bH$  of sparse \emph{circulant} matrices (or simply circulants) of the same size over the field GF($q$) where $q $ is a power of prime.  If the array $\bH $, viewed as a matrix, has constant column weight $\gamma$  and constant row weight $\rho $, the code given by the null space of $\bH$  is said to be ($\gamma $,$\rho$)-regular, otherwise it is said to be irregular.

   In almost all of the proposed constructions of LDPC codes, the following constraint on the rows and columns of the parity-check matrix $ \bH$  is imposed:  \emph{no two rows (or two columns) can have more than one place where they both have identical non-zero components}. This constraint on the rows and columns of $\bH$  is referred to as the  \emph{row-column (RC)-constraint}. This RC-constraint ensures that the Tanner graph [9] of the LDPC code given by the null space of $\bH$  has a girth of at least 6 and that the minimum distance of the code, if ($\gamma $,$\rho $)-regular, is at least $\gamma  + 1 $ [10], [11].  The distance bound is tight for regular LDPC codes whose parity-check matrices have large column weights and  \emph{row redundancies}, such as the algebraic LDPC codes constructed using finite fields, finite geometries and combinatorial designs. A parity-check matrix $\bH$  that satisfies the RC-constraint is called an RC-constrained parity-check matrix and the code given by its null space is called an RC-constrained LDPC code.

   The overall performance of an LDPC code with iterative decoding based on belief propagation is measured by: (1) its bit and block error performance (i.e., how close it performs to the Shannon limit or sphere packing bound); 2) the rate of decoding convergence (i.e., how fast the decoding converges to a valid codeword); (3) its error-floor; and (4) how efficient it can be encoded and decoded. Extensive studies and simulation results show that the performance of an LDPC code is determined by a number of structural properties of the code collectively: (1) minimum distance (or minimum weight); (2) girth and cycle distribution of its Tanner graph; (3) degree distributions of variable- and check-nodes of its Tanner graph; (4) trapping set distribution of its Tanner graph; (5) row redundancy of its parity-check matrix; and (6) other unknown structures. No single structural property dominates the performance of a code. It is still unknown how the code performance depends on the above structural properties analytically as a function. However, some general information is known how to design (or construct) LDPC codes that perform well.  Recently, it was proved that for an RC-constrained ($\gamma$,$\rho$)-regular LDPC code, its Tanner graph contains no trapping set of size $\kappa\leq \gamma$ with the number of odd degree check-nodes smaller than $\gamma$ [12].  Several classes of algebraic LDPC codes were proved that they do not contain harmful trapping sets of sizes smaller than their minimum distances. Consequently, the error-floors of these codes are primarily dominated by their minimum distances.

   For a code with a given rate to perform close to the Shannon capacity (or its threshold) in the waterfall region, the degree distributions of the variable- and check-nodes of its Tanner graph must be properly designed (say, based on density evolution [13]).  For a code to have low error-floor, it must have a relatively large minimum distance and no harmful trapping sets with sizes smaller than its minimum distance.  In this case, the error-floor is dominated by its minimum distance.  Furthermore, the error-floor performance of a code also depends on the girth of the code's Tanner graph.  In general, a girth of 6 is enough if the code has large minimum distance and no small trapping sets.  For the decoding of a code to converge fast, besides requiring no harmful small trapping sets and relatively large minimum distance, large row redundancy (large number of dependent rows) of its parity-check matrix helps.  Extensive simulation results show that the decoding of a code converges very fast, if its parity-check matrix has a large row redundancy.  Cyclic and QC-LDPC codes constructed using finite fields, finite geometries and combinatorial designs do have large row redundancies in their parity-check matrices.  Iterative decoding of these codes does converge very fast.  For efficient encoding and decoding hardware implementation of an LDPC code, quasi-cyclic or cyclic structure is desirable.  How to design or construct an LDPC code with the above good structures is an unsolved but challenging problem.

   In general, QC-LDPC codes (regular or irregular) given by the null spaces of arrays of circulants constructed algebraically based on finite fields, finite geometries and combinatorial designs [10], [12], [14]-[24] do have a good balance in terms of minimum distance, trapping set structure, row redundancy, and girth.  Masking the parity-check array of a QC-LDPC code based on well designed degree distributions of the code's Tanner graph also provides good error performance in the waterfall region as shown in [20].

   Recent development in QC-LDPC codes (QC codes in general) is the introduction of a matrix-theoretic approach for studying these code based on matrix transformation via  \emph{Fourier transforms} [25], [26].  This approach is amicable to the analysis and construction of QC-LDPC codes.  In Fourier transform domain, the parity-check matrix of a QC-LDPC code, as an array of circulants, is specified by a set of  \emph{base matrices} (or a single base matrix) over a finite field that satisfies  certain constraints.  Based on these base matrices, an RC-constrained array of sparse circulants can be easily constructed.  The null space of this RC-constrained array then gives an RC-constrained QC-LDPC code whose Tanner graph has a girth of at least 6.  From these base matrices, it is quite easy to analyze the rank of the parity-check array of the code and to derive the necessary and sufficient condition for the code's Tanner graph to have a given girth.  The Fourier transform approach put all the algebraic constructions of QC-LDPC codes developed in [17], [20], [21]-[24], [25] under a single framework.

   Although many types of algebraic constructions of QC-LDPC codes have been proposed and some rank expressions have been given [17], [20], [21]-[24], [25], there is still a lack of general algorithms and guideline to construct QC-LDPC codes with large redundancy, and the existing rank analysis and expressions are only applicable to some specific cases.   

   In this paper, we follow the Fourier transform approach presented in [25] to expand the analysis and construction of new QC-LDPC codes.  Analysis includes the rank and row redundancy of an array of circulants whose null space gives a QC-LDPC code.  A recursive algorithm for computing the rank or the row redundancy of a parity-check array in terms of its base matrices in Fourier transform domain is developed. Tight upper and lower bounds on the rank and row redundancy of an array of circulants are derived.  In special cases, combinatorial expressions for the exact ranks are obtained. New constructions of algebraic QC-LDPC codes in the Fourier transform domain with large redundancy are given, and the simulation result demonstrate that the constructed QC-LDPC codes outperform the corresponding random LDPC codes .  Ranks and row redundancies of the parity-check arrays of some known QC-LDPC codes are further investigated.

   The organization of the rest of this paper is as following: First, we present the characterization of QC codes, binary QC-LDPC codes and nonbinary QC-LDPC codes in the Fourier transform domain in Section II, Section III and Section IV, respectively. Then, we analyze ranks and row redundancies of QC-LDPC codes in terms of transform domain in Section V. A tight upper bound on ranks and a tight lower bound on row redundancies are given in this section. Later, we explain the reason why row redundancies can increase the performance of message-passing algorithms. In Section VI, we further analyze the rank for two types of well-known LDPC codes constructed based on Vandermonde Matrices and Latin Squares and show that the bound derived in Section V we construct a class of RC-constrained QC-LDPC codes based on random partitions of finite fields. In Section VII, we propose several types of constructions of QC-LDPC codes with large row redundancies which outperform random LDPC codes. The paper is summarized in the last section.

\section{Characterization of QC Codes in the Fourier Transform Domain}

   In this paper, we consider only QC-LDPC codes constructed from finite fields of characteristic of 2.  In this and next sections, we give a review of characterization of QC codes in Fourier transform domain presented in [25]. Some new interpretations and extensions are given.

\subsection{ Matrix Transformation}

   Let GF($2^r$) be a finite field with $2^r$ elements which is an extension field of the binary field GF(2). Let $\alpha$  be a primitive element of GF($2^r$).  Then, the powers of $\alpha $, $\alpha ^{-\infty} $,$ \alpha ^0 = 1$, $\alpha$, $\alpha ^2$, . . . , $\alpha ^{2^r - 2}$, give all the elements of  GF($2^r$) and $\alpha ^{2^r -1} = 1$.

   Let $q = 2^r$ and $e = 2^r - 1$. Let $\ba = (a_0, a_1, . . . , a_{e-1})$ be an $e$-tuple (or vector) over GF(2). Its  \emph{Fourier transform} [25], [27] denoted by ${\cal {\cal F}}[\ba]$, is given by the $e$-tuple $\bb = (b_0, b_1, ... , b_{e-1})$ over GF($q$) whose $t$-th component, $b_t$, for $0 \leq  t < e$, is given by
\begin{equation}              
                b_t = a_0 + a_1\alpha ^t + a_2\alpha ^{2t} + \cdots + a_{e-1}\alpha ^{(e-1)t}. 
\end{equation}                      
The vector $\ba$, which is called the  \emph{inverse Fourier transform} of the vector $\bb$, denoted by $\ba = {\cal {\cal F}}^{-1}[\bb]$, can be retrieved using the following equation:
\begin{equation}
               a_l = b_0 + b_1\alpha ^{-l} + b_2\alpha ^{-2l} + \cdots + b_{e-1}\alpha ^{-(e-1)l} 
\end{equation}
for $ 0 \leq  l < e$.

   An $e\times e$ matrix of over a field is called a circulant if every row is a  \emph{cyclic-shift} (one place to right) of the row above it and the first row is the cyclic-shift of the last row.  A circulant is uniquely specified by its top row which is called the \emph{generator} of the circulant.

   Let $\bA  = [a_{ij}]$, $0 \leq  i, j < e$, be an $e\times e$ circulant over GF(2). Then, we write $\bA  = circ(a_0, a_1, . . . , a_{e-1})$, where $(a_0, a_1, . . . , a_{e-1})$ is the generator of $\bA $. Define two $e\times e$ matrices over GF($q$) as follows: $\bV = [\alpha ^{-ij}]$,  $0 \leq  i, j < e$ and $\bV ^{-1} = [\alpha ^{ij}]$,  $0 \leq  i, j <e$. Both matrices, $\bV $ and $\bV ^{-1}$, are known as  \emph{Vandermonde} matrices [27], [28] and they are non-singular. Furthermore $\bV\bV^{-1} = \bI$, where $\bI $ is an $e\times e$ identity matrix.  Hence, $\bV ^{-1}$ is the inverse of $\bV  $ and vice versa. Taking the matrix product $\bV\bA\bV^{-1}$, we obtain the following $e\times e$ diagonal matrix over GF($2^r$),
\begin{equation}
           {\bA}^{\cal {\cal F}} =  \bV\bA\bV^{-1}  = diag(b_0, b_1, ..., b_{e-1}).
\end{equation}
where the diagonal vector $(b_0, b_1, . . . , b_{e-1})$ is the Fourier transform of the  generator  $(a_0, a_1, . . . , a_{e-1})$ of the circulant of $\bA $.  The diagonal matrix $\bA ^{\cal {\cal F}} = \bV\bA\bV^{-1}$ is referred to as the Fourier transform of the circulant $\bA $.  In the rest of the paper, we only consider circulants over GF($q$) of size $e\times e$ with $q = 2^r$  and $e = 2^r - 1$.


   Since $(a_0, a_1, . . . , a_{e-1})$ is an $e$-tuple over GF(2), the components must satisfy the following constraint [25]:
\begin{equation}
                                              b_{(2t)_e} = b_t ^2
\end{equation}
for $0 \leq  t < e$, where $(2t)_e$ denotes the nonnegative integer less than $e $ and congruent to $2t \mbox{ modulo } e$.  This condition is known as the  \emph{conjugacy constraint}. Conversely, if an $e$-tuple $(b_0, b_1, . . . , b_{e-1})$ over GF($q$) satisfies the conjugacy constraint, its inverse Fourier transform gives an $e$-tuple $(a_0, a_1, . . . , a_{e-1})$ over GF(2).

   Let $m $ and $n $ be two positive integers. Let $\bH  = [\bA_{i,j}]$ $0 \leq  i < m$, $0 \leq  j < n$ be an $m\times n$ array of $e\times e$ circulants $\bA_{i,j}$ over GF(2). For $0 \leq  i < m$, $0 \leq  j < n$, let $(a_{i,j,0}, a_{i,j,1}, . . . , a_{i,j,e-1})$ be the generator of the circulant $\bA_{i,j}$. Next, we define two diagonal arrays of $e\times e$ Vandermonde matrices $\bV  $ and $\bV ^{-1}$ as follows:
\begin{equation}
\begin{array}{ccc}
                         \Omega (m) &=& diag( \underbrace{\bV , \bV , . . . , \bV}_{m}  ),   \\
            \end{array}         \end{equation}
\begin{equation}
\begin{array}{ccc}
                         \Omega ^{-1}(n) &=& diag(\underbrace{\bV ^{-1}, \bV ^{-1}, . . . , \bV ^{-1}}_{n}), \\
              \end{array}       
\end{equation}
where $\Omega (m)$ is an $m\times m$ diagonal array of the Vandermode matrices $\bV $'s and $\Omega ^{-1}(n)$ is an $n\times n$ diagonal array of Vandermonde matrices $\bV ^{-1}$'s.  Then the Fourier transform of $\bH$  is given as follows:
 \begin{equation}
\begin{array}{lll}
             \bH^{\cal {\cal F}} &= &\Omega(m) \bH \Omega ^{-1}(n)   \\
 &=&                               
\left[\begin{array}{llll}
    \bA_{0,0}^{\cal F}   &     \bA_{0,1}^{\cal F} &  ... &     \bA_{0,n-1}^{\cal F} \\
         \bA_{1,0}^{\cal F} &   \bA_{1,1}^{\cal F} &  ...&     \bA_{1,n-1}^{\cal F}    \\
          \vdots& &\ddots & \vdots\\
                                       \bA_{m-1,0}^{\cal F} &    \bA_{m-1,1}^{\cal F}&  ...&  \bA_{m-1,n-1}^{\cal F}   
\end{array}\right].
                          \end{array}
\end{equation}
where $\bA_{i,j}^{\cal F} = \bV\bA_{i,j}\bV^{-1}$, $ 0 \leq  i < m$, $0 \leq  j < n$, is an $e\times e$ diagonal matrix over GF($2^r$) with diagonal vector $(b_{i,j,0}, b_{i,j,1}, . . . , b_{i,j,e-1})$ which is the Fourier transform of the generator $(a_{i,j,0}, a_{i,j,1}, . . . , a_{i,j,e-1})$ of $\bA_{i,j}$.

   The array $\bH ^{\cal F}$ is an $m e\times ne$ matrix over GF($q$).  Label the rows and columns of $\bH ^{\cal F}$ from $0$ to $me - 1$ and 0 to $ne - 1$, respectively.    Define the following index sequences: for $0 \leq  i, j < e$, 
\begin{equation}
                  \pi _{row,i} = [i, e+i, ..., (m-1)e+i],
\end{equation}
and 
\begin{equation}
                  \pi _{col,j} = [j, e+j, ..., (n-1)e+j].
                  \end{equation}
Let
\begin{equation}
                         \pi_{ row} = [\pi _{row,0}, \pi_{ row,1}, . . . , \pi_{ row,2^r-2}],
\end{equation}
and
 \begin{equation}
                         \pi _{col} = [\pi _{col,0}, \pi_{ col,1}, . . . , \pi_{ col, 2^r-2}].
\end{equation}
Then $\pi_{row}$ gives a permutation of the indices (labels) of the rows of $\bH ^{\cal F}$ and $\pi _{col}$ gives a permutation of the indices of columns of $\bH ^{\cal F}$.

   Suppose we first permute the rows of $\bH ^{\cal F}$ based on $\pi_{row}$ and then the columns based on $\pi _{col}$. These row and column permutations result in the following $e\times e$ diagonal array of $m \times n$ matrices over GF($q$),
\begin{equation}
\begin{array}{ccl}
                 \bH ^{{\cal F},\pi } &=& diag(\bB_0, \bB_1, . . . , \bB_{2^r-2})\\
&=& \left[\begin{array}{ccccc}
                                  \bB_0 &   {\bf O} &      {\bf O}  &.  .  .& {\bf O} \\
                                 {\bf O} &\bB_1   &  {\bf O}  &  .  .  .&  {\bf O} \\
                                  \vdots &  & & \ddots & \vdots\\
                                  {\bf O} &  {\bf O} &  {\bf O} &  .  .  .&   \bB_{2^r-2}
\end{array}\right],
\end{array}
\end{equation}
where $\pi  = (\pi_{ row}, \pi_{ col})$ denotes the combination of the row and column permutations, $\pi _{row}$ and $\pi_{ col}$. The $m \times n$ matrices $\bB_i$'s on the diagonal of the array $\bH ^{F,\pi }$ satisfy the conjugacy constraint.  To specify this constraint, we introduce the concept of Hadamard product.

    Let $\bB = [b_{i,j}]$ and $\bC = [c_{i,j}]$ be two matrices of the same size. The Hadamard product of $\bB$ and $\bC$ [29], denoted by $\bB\circ \bC$, is defined as their  \emph{element-wise} product, i.e., $\bB\circ \bC = [b_{i,j}c_{i,j}]$. The Hadamard product of $t$ copies of the matrix $\bB$, where $t$ is a nonnegative integer, denoted by $\bB^{\circ t}$, is $\bB^{\circ t} = [b_{i,j}^t]$ which is referred to as the $t$-th  \emph{Hadamard power} of $\bB$. We allow $t$ to equal 0 and in this case $b_{i,j}^0 = 1$ if  $b_{i,j}$ is a nonzero element in GF($2^r$) and $b_{i,j}^0 = 0$ if $b_{i,j} = 0$.

   For an array $\bH$  of circulants and zero matrices over GF(2), the matrices on the main diagonal of the array $\bH ^{{\cal F},\pi}$ satisfy the conjugacy constraint [25],
\begin{equation}
                                         \bB_{(2t)_e} = \bB_t^{\circ 2},
\end{equation}
i.e., the entry at the location $(i,j)$ of $\bB_{(2t)_e}$ is the square of the entry at the location $(i,j)$ of $\bB_t$.  We call the matrix $\bB_{(2t)_e}$  a \emph{conjugate matrix} of $\bB_t$.

   Conversely, if the matrices on a diagonal array of the form given by (12) satisfy the conjugacy constraint given by (13), then the array obtained by taking inverse row and column permutations and inverse Fourier transform, we obtain an array of circulants over GF(2).

   The transformation from $\bH$  to $\bH ^{{\cal F},\pi }$ through $\bH ^{\cal F}$ is reversible. Given an $e\times e$ diagonal array ${\hat \bH} = \bH ^{{\cal F},\pi } = diag(\bB_0, \bB_1, . . . , \bB_{e-1})$ of $m\times n$ matrices over GF($q$), one can perform inverse permutation $\pi ^{-1} = (\pi_{ row}^{-1}, \pi _{col}^{-1})$ on the rows and columns of ${\hat \bH}$ to obtain an $m\times n$ array ${\hat \bH}^{\pi ^{-1}}$ of $e\times e$ diagonal matrices $\bA_{i,j}^{\cal F}$.  Next, perform inverse Fourier transform on ${\hat \bH}^{\pi ^{-1}}$, i.e., replacing each diagonal matrix $\bA_{i,j}^{\cal F}$ in ${\hat \bH}^{\pi ^{-1}}$ by an $e\times e$ circulant whose first row is the inverse Fourier transform of the diagonal vector of the diagonal matrix $\bA_{i,j}^{\cal F}$. This results in an $m\times n$ array ${\hat \bH}^{\pi ^{-1},{\cal F}^{-1}} = \bH$  of $e\times e$ circulants over GF(2). Thus, we have a  \emph{one-to-one} correspondence between an array of circulants over GF(2) and a diagonal array of matrices over GF($q$).

   The transformation from $\bH $ to $\bH ^{{\cal F},\pi }$ preserves the rank of the matrices. Let $rank(\bM)$ denote the rank of a matrix $\bM $ over a finite field.  Since $\bH ^{\cal F,\pi}  = diag(\bB_0,\bB_1, . . . ,\bB_{e-1})$, then
\begin{equation}\label{brank}
                    rank(\bH) = rank(\bB_{ 0} ) + rank(\bB) + rank(\bB_{ 2}) + \cdots + rank(\bB_{(e -1)}),
                    \end{equation}
In a latter section, we develop a recursive algorithm for computing the rank of $ \bH $, $rank(\bH )$, based on the conjugacy constraint on matrices, $\bB_0, \bB_1, . . . , \bB_{e-1}$, given by (13).

\subsection{Characterization of Binary QC Codes in Fourier Transform Domain}

   Consider a binary QC code ${\cal C}_{qc}$ given by the null space of an $m\times n$ array $\bH  = [\bA_{i,j}]$ $0 \leq  i < m$, $0 \leq  j < n$ of $e\times e$  circulant matrices $\bA_{i,j}$ over GF(2).  $\bH$  is an $me\times ne$ matrix over GF(2).  The one-to-one correspondence between arrays $\bH  = [\bA_{i,j} ]$ of circulant matrices and diagonal arrays $\bH ^{{\cal F},\pi } = diag(\bB_0, \bB _1, . . . , \bB _{e-1})$ of matrices and the conjugacy constraint on the matrices $\bB _0, \bB _1, . . . , \bB _{e-1}$ on the diagonal of $\bH ^{{\cal F},\pi }$ give the basis for studying QC codes in Fourier transform domain.

   Partition the set ${\cal E} = \{0, 1, . . . , e -1\}$ of integers into  \emph{cyclotomic cosets} of $2$ modulo $e $ [25], [30] where $e = 2^r - 1$.  Let $t $ be an integer in $\cal E$. The  \emph{cyclotomic coset} containing $t $ is
\begin{equation}
                                                         {\bf Z}_t = \{t, (2 t)_e, (2^2t)_e, . . . , (2^{c_t -1}t) _e \},
\end{equation}
where $c_t$ is the smallest positive integer satisfying $2^{c_t} t \equiv t \mod e$. Each coset has a smallest member which we call the  \emph{coset representative}. The conjugacy constraint given by (4) constrains the components of the Fourier transform of a binary vector whose indices are in the same cyclotomic coset. All these components are  \emph{powers} of the component whose index is the coset representative.

  It follows from the conjugacy constraint on the matrices $\bB _0, \bB _1, . . . , \bB _{e-1}$ in the diagonal array $\bH ^{{\cal F},\pi }$ given by (13) that all matrices $\bB_ t$ whose indices are in the same cyclotomic coset are determined by the matrix whose index is the coset representative. The matrices with indices in the same cyclotomic coset modulo $e$ are conjugate matrices which form a  \emph{conjugate class}.  Given one matrix in a conjugate class, we can determine all the other conjugate matrices in the same class. Consequently, the binary parity-check array $\bH$  is determined by the matrices $\bB _t$'s for which the $t$'s are coset representatives of all the distinct cyclotomic cosets. In particular, $\bH$  is specified by a number of matrices $\bB _t$ equal to the number of cyclotomic cosets of 2 modulo $e$. Therefore, the construction of an $m\times n$ array of $e\times e$ circulants over GF(2) consists of the following steps:

 \begin{enumerate}
 \item   Determine the cyclotomic cosets of 2 modulo $e$.  Let ${\bf Z}_0, {\bf Z}_1, . . . , {\bf Z}_{\lambda -1}$ be all the cyclotomic cosets modulo of 2  modulo $e$, where ${\bf Z}_0 = \{0\} $ and $\lambda$  is the number of cyclotomic cosets.  Let $t_0 = 0, t_1, . . . , t_{\lambda -1}$ the coset    representatives of ${\bf Z}_0, {\bf Z}_1, . . . , {\bf Z}_{\lambda -1}$.

\item   Choose $\lambda$  $m\times n$ matrices $\bB_{ t_0}, \bB _{t_1}, . . . , \bB _{t_{\lambda -1}}$ over GF($q$) with $q = 2^r$.

\item   For each $\bB_{ t_i}$, $0 \leq  i < \lambda $, we form all its conjugate matrices. This gives $e $ matrices $\bB _0, \bB _1, . . . , \bB _{e-1}$ of size $m\times n$.

\item  Form the $e\times e$ diagonal array $\bH ^{{\cal F},\pi } = diag(\bB _0, \bB _1, . . . , \bB _{e-1})$.

\item    Performing inverse permutations $\pi _{row}^{-1}$ and $\pi _{col}^{-1}$ on the rows and columns of the array $\bH ^{{\cal F},\pi }$ (as an
     $  me\times me $ matrix over GF($q$)), we obtain an $m\times n$ array $\bH ^{\cal {\cal F}}$ of $e\times e$ diagonal matrices over GF($q$).

\item    Performing the inverse Fourier transform ${\cal F}^{-1}$ on the array $\bH ^{\cal {\cal F}}$, we obtain an array $m\times n$ array of circulant over GF(2).

 \end{enumerate}

The null space of $\bH$  gives a QC code ${\cal C}_{qc}$.  Therefore, the construction of a binary QC code is determined by the choice of the base matrices $\bB_{ t_0}, \bB_{ t_1}, . . . , \bB _{t_{\lambda -1}}$.

   If $\bH$  is an array of sparse circulants over GF(2), then the null space of $\bH$  gives a QC-LDPC code ${\cal C}_{qc}$.  As an $me\times ne$ matrix, if $\bH$  satisfies the RC-constraint, then the Tanner graph of the QC-LDPC code ${\cal C}_{qc}$ given by the null space of $\bH$  has a girth at least 6.  If $\bH$  is a regular matrix with column weight $\gamma $, then the minimum distance of ${\cal C}_{qc}$ is at least $\gamma  + 1$.

\section{ Characterization of a Class of Binary RC-Constrained QC-LDPC Codes in Fourier Transform Domain}

    Typically, in most constructions of parity-check matrices of QC-LDPC codes, each circulant is either a zero matrix (ZM) or a  \emph{circulant permutation matrix} (CPM), i.e., a circulant with exactly one non-zero entry in each row and each column and this entry is 1.

   If $\bH$  is an $m\times n$ array of CPMs and/or zero matrices (ZMs) of size $e\times e$, the conjugacy constraint of (13) becomes the following constraint [25]:
\begin{equation}
                                            \bB _t = \bB _1 ^{\circ t}                   
\end{equation}
for $0 \leq  t < e$, i.e., $\bB_t$ is the $t$-th Hadamard power of $\bB_1 $.   In this case, the array $\bH ^{{\cal F},\pi }$ given by (12) is uniquely specified by the matrix $\bB _1$. As a result, we could remove the subscript ``1'' from $\bB _1$ and use $\bB$  for $\bB _1$.  Then, the array $\bH ^{{\cal F},\pi }$ has the following form [25]:
\begin{equation}
\begin{array}{lll}
            \bH ^{{\cal F},\pi} & = &diag(\bB ^{\circ 0}, \bB ^{\circ 1}, . . . , \bB ^{\circ (2^r-2)} )\\
                   &= & \left[\begin{array}{ccccc}
                                \bB ^{\circ}&  \bO &        \bO &       .  .  .&        \bO \\
                                   \bO &   \bB ^{\circ 1} &      \bO &     .  .  .  &      \bO \\
                                   \vdots&               &                   &         \ddots & \vdots\\
                                  \bO  &   \bO  & \bO &          . . . &  \bB ^{\circ (2^r-2)}
                                  \end{array}\right].
\end{array}
\end{equation}
   The result given by (17) actually says that, in the Fourier transform domain, any array $\bH $ of CPMs and/or ZMs is completely specified by a matrix $\bB$  over GF($q$) with $q = 2^r$.  Any $m\times n$ matrix over GF($q$) can be used as the $\bB$  matrix.

   To contruct a QC-LDPC code, we begin with an appropriately chosen $m\times n$ matrix $ \bB$  over a finite field GF($q$).  Form an $e\times e$ diagonal array ${\hat \bH} = diag(\bB ^{\circ 0}, \bB ^{\circ 1}, . . . , \bB ^{\circ (2^r-2)})$ of the form (17).  Next, we apply the permutation $\pi ^{-1} = (\pi_{ row}^{-1}, \pi _{col}^{-1}) $ on the rows and columns of $ {\hat \bH}$ to obtain an $m\times n$ array ${\hat \bH}^{\pi ^{-1}}= [\bA_{i,j}^{\cal {\cal F}}]$  of diagonal matrices $\bA_{i,j}^{\cal {\cal F}}$ of size of $e\times e$.  Then, we take the inverse Fourier transform of ${\hat \bH}^{\pi ^{-1}}$ to obtain an $m\times n$ array ${\hat \bH}^{\pi ^{-1},{\cal F}^{-1}} = \bH  = [\bA_{i,j}]$ $0 \leq  i < m, 0 \leq  j < n$ of CPMs and/or ZMs of size $e\times e$.  $\bH$  is an $me\times ne$ matrix over GF(2).  For $r \geq  3$, $\bH$  is a low-density matrix. The null space of $\bH$  gives a QC-LDPC code ${\cal C}_{qc}$.  Since the array $\bH$  is constructed from $\bB $, we call $\bB  $ the  \emph{base matrix} for code construction.

   If the base matrix $\bB$  satisfies the condition given by the following theorem, then the parity-check matrix $\bH $ of the QC-LDPC code ${\cal C}_{qc}$ satisfies the RC-constraint and its Tanner graph has a girth at least 6.  We will state the theorem without a proof.  A  proof can be found in [25].

\begin{theorem}A  necessary and sufficient condition for an array $\bH $ of CPMs and/or ZMs to satisfy the RC-constraint is that every $2\times 2$ submatrix in the base matrix $\bB$  contains at least one zero entry or is non-singular.
\end{theorem}

The necessary and sufficient condition on a base matrix given in Theorem 1 is called the $2\times 2$ submatrix (SM)-constraint. $\bA$  base matrix $\bB$  that satisfies the $2\times 2$ SM-constraint is called a $2\times 2$ SM-constrained base matrix.

   Next, we show that construction of an RC-constrained LDPC matrix of a QC-LDPC code which consists of an array of CPMs and/or ZMs can be carried out directly from a $2\times 2$ SM-constrained base matrix $\bB$  without forming the array ${\hat \bH}= diag(\bB ^{\circ 0}, \bB ^{\circ 1}, . . . , \bB ^{\circ (e -1)})$, taking the inverse row and column permutations and the inverse Fourier transform.

   Consider an $e\times e$ CPM $\bA  = circ(a_0, a_1, . . . , a_{e-1})$ over GF(2) with generator $(a_0, a_1, . . . , a_{e-1})$ which contains a single 1-component.  Suppose the single 1-component of $(a_0, a_1, . . . , a_{e-1})$ is at the $k$th position, i.e., $a_k= 1$ and $ a_t = 0$ for $t \neq  k$.   It follows from (1) and (2) that the diagonal vector $(b_0, b_1, . . . , b_{2^r-2})$ of the Fourier transform $\bA ^{\cal {\cal F}}$ of $\bA$  is
\begin{equation}
                   (b_0,b_1, . . . , b_{2^r-2}) = (\alpha ^0, \alpha ^k, \alpha ^{2k}, . . . , \alpha ^{(e-1)k}),          
\end{equation}
which consists of $e$ consecutive powers of $\alpha ^k$. Conversely, if a diagonal matrix over GF($q$) with diagonal vector $(\alpha ^0, \alpha ^k, \alpha ^{2k}, . . . , \alpha ^{(e-1)k})$, then its inverse Fourier transform is an $e\times e$ CPM whose generator has its single 1-component at the position $k$.

   Let $\bB  = [b_{i,j}]$, $0 \leq  i < m, 0 \leq  j < n$, be the chosen base matrix for code construction. Construct the diagonal array $\bH ^{{\cal F},\pi}  = diag(\bB ^{\circ 0}, \bB ^{\circ 1}, . . . , \bB ^{\circ (e-1)})$ given by (17), where $\bB ^{\circ t} = [b_{i,j}^t]$, $0 \leq  i < m$, $0 \leq  j < n$, for $0 \leq  t <e$. Applying the permutation $\pi ^{-1} = (\pi _{row}^{-1}, \pi _{col}^{-1})$ on the rows and columns of $\bH ^{{\cal F},\pi}$, we obtain the array $\bH ^{\cal {\cal F}} = [\bA_{i,j}^{\cal {\cal F}}]$  of diagonal matrices $\bA_{i,j}^{\cal {\cal F}}$ .  For $0 \leq  i < m$, $0 \leq  j < n$, we find that the diagonal vector of $\bA_{i,j}^{\cal {\cal F}}$ is $(1, b_{i,j}, b_{i,j}^2,. . . , b_{i,j}^{e-1})$.  If $b_{i,j} = \alpha ^k$, then $ \bA_{i,j}$, the inverse Fourier transform of $\bA_{i,j}^{\cal {\cal F}}$, is an $e\times e$ CPM whose generator $(a_0, a_1, . . . , a_{e-1})$ has its single 1-component at the position $k$.

   Based on the above analysis, construction of an RC-constrained low-density parity-check array $\bH$  of CPMs and/or ZMs can be constructed directly from a chosen base matrix $\bB  = [b_{i,j}]$, $0 \leq  i < m$, $0 \leq  j < n$ which is $2\times 2$ SM-constrained.  This is carried out as follows: (1) if $b_{i,j}$ is an nonzero element in GF($2^r$) and $b_{i,j} = \alpha ^k$ with $0 \leq  k < e=2^r - 1$, then we replace $b_{i,j}$ by an $e\times e$ CPM whose generator has its single 1-component at position $k$; and (2) if $b_{i,j} = 0$, then we replace $b_{i,j}$ by an $e\times e$ ZM.  This gives the RC-constrained array $\bH$  of CPMs and/or ZMs corresponding to the chosen base matrix $\bB$  that satisfies the $2\times 2$ SM-constraint. Then, the null space of $\bH$  gives an RC-constrained QC-LDPC code whose Tanner graph has a girth at least 6. The above replacement of an entry $b_{i,j}$ in a base matrix $\bB$  by either an $e\times e$ CPM or an $e\times e$ ZM is referred to as the  \emph{$e$-fold matrix dispersion} of $b_{i,j}$.  The array $\bH$  is called the  \emph{$e$-fold array dispersion} of $\bB$  [20].

   It is clear that the transpose $\bB ^{\sf T}$ of a $2\times 2$ SM-constrained base matrix $\bB$  also satisfies the $2\times 2$ SM-constraint and hence it can be used as a base matrix to form an RC-constrained array of CPMs and/or ZMs whose null space gives an RC-constrained QC-LDPC code.  If $\bH$  is an RC-constrained array of CPMs and/or ZMs constructed from $\bB $, then the RC-constrained array constructed from $\bB ^{\sf T}$ is the transpose $\bH ^{\sf T}$ of $\bH $.

   The above construction puts all the constructions of QC-LDPC codes based on finite fields given in [17], [19], [21]-[24], [25] under a single framework. In these papers, the base matrices are constructed based on finite fields and combinatorial designs.

   Consider an RC-constrained QC-LDPC code ${\cal C}_{qc}$ given by an $m\times n$ array $\bH  = [\bA_{i,j}]$ of $e\times e$ CPM's and/or ZM's which is specified by a $2\times 2$ SM-constrained $m\times n$ base matrix $\bB  = [ b_{i,j}]$.  If $b_{i,j} \neq  0$, $0 \leq  i < m$, $0 \leq  j < n$, then multiplying it by zero results in replacing the CPM $\bA_{i,j}$ by a zero matrix. This procedure, known as  \emph{masking}, was used in previous work to optimize the column and row weights of the parity-check matrices and to reduce the number of short cycles in the Tanner graphs of the constructed codes [20]. This is accomplished by judiciously designing an $m\times n$ binary matrix ${\bf Z} = [z_{i,j}]$, $0 \leq  i < m$, $0 \leq  j < n$, which we call a  \emph{masking matrix}. After masking, we obtain the masked base matrix $\bB _{mask} = {\bf Z}\circ \bB  = [z_{i,j}b_{i,j} ]$, whose $(i, j)$ entry equals $b_{i,j}$ if $z_{i,j} = 1$ and equals zero if $z_{i,j} = 0$.  Performing $e$-fold matrix dispersion of each entry in the masked base matrix $\bB _{mask} $, we obtain a masked array, denoted by $\bH_{mask} $.  The null space of the masked array $\bH_{mask}$ gives a new RC-constrained QC-LDPC code. Masking is an effective technique for construction both regular and irregular QC-LDPC codes [20].

\section{Characterization of a Class of Non-Binary RC-Constrained QC-LDPC Codes in the Fourier Transform Domain}

   In this section, we show that RC-constrained arrays of non-binary CPMs of a special type can also be constructed using the base matrices constructed in Section II.  The null spaces of these arrays give a class of non-binary QC-LDPC codes.

   Again we consider code construction based on fields of characteristic of 2.  Let $\alpha$  be a primitive element of GF($q$) with $q = 2^r$.  Again, let $e = 2^r - 1$. For $0 \leq  k < e$, let $\bP(\alpha ^k)$ be an $e\times e$ matrix over GF($q$) with columns and rows labeled from 0 to $e - 1$ which has the following structures: (1) the top row of $\bP(\alpha ^k)$ has a single nonzero component with value $\alpha ^k$ at the $k$-th position; and (2) every row of $\bP(\alpha ^k)$ is the cyclic-shift (one place to the right) of the row above it multiplied by $\alpha$  and the first row is the cyclic-shift of the last row multiplied by $\alpha $.  This $e\times e$ matrix $\bP(\alpha ^k)$ over GF($2^r$)  is called an $\alpha $-multiplied CPM [31].  There are $e $ such $\alpha $-multiplied CPMs.  For $0 \leq  k < e$, we represent the element $\alpha ^k$ of GF($q$) by the $\alpha $-multiplied CPM $\bP(\alpha ^k)$.  This representation is one-to-one. $\bP(\alpha ^k)$ is referred to as the $\alpha $-multiplied CPM dispersion (or simply dispersion) of $\alpha ^k$.

   Next, we replace each nonzero entry (a power of $\alpha$) of a chosen $2\times 2$ SM-constrained $m\times n$ base matrix $\bB$  by its corresponding $\alpha $-multiplied $e\times e$ CPM and a zero entry by an $e\times e$ ZM.  This results in an $m\times n$ RC-constrained array $\bH _{\alpha}$  of $\alpha $-multiplied CPMs of size $e\times e$ over GF($q$).  It is an $me\times me$ matrix over GF($q$).  Its null space gives a $q$-ary RC-constrained QC-LDPC code whose Tanner graph has a girth of at least 6.

   The array $\bH _{\alpha}$  consists of $n$ column blocks of $\alpha $-multiplied CPMs, denoted $\bH _{\alpha} ^{(0)}, \bH _\alpha ^{(1)}, . . . , \bH _\alpha ^{(n-1)}$. Each column block $\bH _\alpha ^{(j)}$ of $\alpha $-multiplied CPMs with $0 \leq  j < n$ is an $me\times e$ matrix over GF($q$).  Due to the structure of an $\alpha $-multiplied CPM, all the nonzero elements in the $k$-th column of $\bH _\alpha ^{(j)}$ are $\alpha ^k$ for  $0 \leq  k < e$ and $0 \leq  j < n$.  We call $\alpha ^k$ the value of the $k$-th column of $\bH _\alpha ^{(j)}$.  View the overall array $\bH _\alpha$  as an $me\times me$ matrix over GF($q$). If we multiply each column of $\bH _\alpha$  by the multiplicative inverse of its value, we obtain the binary array $\bH$  of CPMs constructed from the base matrix $\bB$  as given in Section III.  Therefore, the rank of $\bH _\alpha$  is the same as that of $\bH $, i.e.,
\begin{equation}
\begin{array}{ccl}
                         rank(\bH _\alpha ) &= &rank(\bH )\\ 
&=& rank(\bB^{\circ 0} ) + rank(\bB) + rank(\bB^{\circ 2}) + \cdots + rank(\bB^{\circ(e -1)}).
\end{array}        
            \end{equation}
   Masking can also be performed on the base matrix $ \bB$  to construct regular or irregular non-binary QC-LDPC codes using $\alpha $-multiplied CPM dispersion of each nonzero entry in the masked base matrix $\bB _{mask}$.

\section{Rank and Row Redundancy Analysis}

   In this section, we give a general analysis of the rank and row redundancy of the parity-check matrix of a QC-LDPC code which is an array of CPMs and/or ZMs in the Fourier Transform domain.  The row redundancy of a matrix is defined as the ratio of the number of redundant rows (or dependent rows) of the matrix to the total number of rows of the matrix.  For an algebraic LDPC code, large row redundancy speeds up the rate of decoding convergence, i.e., requiring smaller number of iterations for the decoder to converge to codeword than other types of LDPC codes.

\subsection{Rank Analysis}

   Consider a binary QC-LDPC code ${\cal C}_{qc}$ given by the null space an $m\times n$ array $\bH  = [\bA _{i,j}]$ ,$ 0 \leq  i < m$, $0 \leq  j < n$, of $e\times e$ CPM's and/or ZMs which is the array dispersion of an $m\times n$ base matrix $\bB$  over GF($q$) where $q = 2^r$ and $e  = 2^r - 1$. The rank of $\bH$  is given by (19). If we can determine the rank of each Hadamard power of the base matrix $\bB $, then we can determine the rank of the parity-check matrix.

   
If $\bH$ is an array  of circulants and ZMs over GF(2), based on the conjugacy constraint specified by (13), it can be readily proved by induction that for any integers $t>0$,
\begin{equation}\label{induc}     
\bB_{(2^it)_e} = \bB_t^{\circ 2^i},
\end{equation}
From (\ref{induc}) we can group the $e $ matrices $\bB_0, \bB_1, . . . , \bB_{e-1}$  into conjugacy classes.  Let $\lambda$  be the number of distinct conjugacy classes and $\Psi _0, \Psi _1, . . . , \Psi _{\lambda -1}$ denote these classes, where $\Psi _0 $ contains only the matrix $\bB _0$ and $\Psi _1$ contains $\bB_1$  and its conjugate matrices.  For $ 0 \leq  i <\lambda $, let $c_i $ be the number of matrices in the conjugacy class $\Psi _i$, where $c_i$ is the smallest nonnegative integer such that $(2^{c_i} t)_e = t$.  Suppose $\bB_{t_i}$ is member matrix in $\Psi _i$, then it follows from (\ref{induc}) that
\[ 
\Psi_i = \{\bB_{t_i}, \bB_{t_i}^{\circ 2}, . . . , \bB_{t_i}^{\circ 2^{c_i-1}}\}=\{\bB_{(t_i)_e},\bB_{(2t_i)_e},...,\bB_{(2^{c_i-1}t_i)_e}\}.
\]
The subscripts of the conjugate matrices in $\Psi _i$ actually form the \emph{cyclotomic coset} ${\bf Z}_i = \{t_i, 2t_i, . . . , 2^{c_i-1} t_i \}$  modulo $e$.  It is clear that for $i = 0$, we have $t_i = 0$ and $c_i = 1$. The matrix in $\Psi _i$ with the smallest power is called the representative of the conjugacy class $\Psi _i$. The following theorem shows that matrices in the same conjugacy class have the same rank.   

\begin{theorem}
Let $\bB$  be an $m\times n$ matrix over GF($q$). For any nonnegative integer $t \leq  r$, the matrix $\bB ^{\circ 2^t} $ (the $2^t$-th Hadamard power of $\bB$) has the same rank as $\bB$.
\end{theorem}
 \begin{proof}
Let $\mu$  be the rank of $\bB $. Let $\bb_{i_1}, \bb_{i_2}, . . . , \bb_{i_\mu}$  be a set of linearly independent rows of $\bB$  where $0 \leq  i_1, i_2, . . . , i_\mu < m$.  Let $a_{i_1}, a_{i_2}, . . . , a_{i_\mu}$  be any set  of $\mu$  elements in GF($q$), not all zero.  Then 
\[
                        \sum\limits^{\mu}_{l=1} a_{i_l} \bb_{i_l}   \neq  0.
\]
Raising the above sum vector to the power $2^t$, since the characteristic of the field GF($q$) is 2, we have
\begin{eqnarray}
                       (  \sum\limits^{\mu}_{l=1} a_{i_l} \bb_{i_l}  )^{\circ 2^t}  &=&                        \sum\limits^{\mu}_{l=1} (   a_{i_l} \bb_{i_l}  )^{\circ 2^t}  \nonumber \\
&=& \sum\limits^{\mu}_{l=1}a_{i_l}^{\circ 2^t}   (\bb_{i_l})^{\circ 2^t}    \neq 0.
\end{eqnarray}
The vectors $(\bb_{i_1})^{\circ 2^t},  (\bb_{i_2})^{\circ 2^t} , . . . , (\bb_{i_\mu} )^{\circ 2^t}$  are $\mu$  rows in the matrix $\bB ^{\circ 2^t}$, the $2^t$-th Hadamard power of $\bB $.  The expression of (20) implies that $(\bb_{i_1})^{\circ 2^t},  (\bb_{i_2})^{\circ 2^t}, . . . , (\bb_{i_\mu})^{\circ 2^t}$  are linearly independent.  This implies that 
\begin{equation} \label{geq}
                                                 rank(\bB ^{\circ 2^t})  \geq   rank(\bB ).
\end{equation}
   Notice that $(\bB ^{\circ 2^t})^{\circ  2^{r-t}} = \bB $.  Let $\mu '$ be the rank of $\bB ^{\circ 2^t}$ and $(\bb_{i_1})^{\circ 2^t},  (\bb_{i_2})^{\circ 2^t }, . . . , (\bb_{i_{\mu '}})^{\circ 2^t}$  be the independent rows of $\bB ^{\circ 2^t}$.  In a similar way, we can show that  $\bb_{i_1}, \bb_{i_2}, . . . , \bb_{i_{\mu'}}$  are linear independent rows of $\bB $.  This implies that
\begin{equation} \label{leq}
  rank(\bB )\geq  rank(\bB ^{\circ 2^{t}}) .   
\end{equation}
The inequalities of (\ref{geq}) and (\ref{leq}) imply that $rank(\bB ^{\circ 2^t})  = rank(\bB )$. 
\end{proof}

It is clear that for $i = 0$, we have $t_i = 0$ and $c_i = 1$. The matrix in $\Psi _i$ with the smallest power is called the representative of the conjugate class $\Psi _i$. For $0 \leq  i < \lambda $, let $\mu _i$ be rank of the matrices in the conjugate class $\Psi _i$, then, it follows from (19) and Theorem 2 that the rank of the parity check matrix $\bH$  of a QC-LDPC code obtained by array dispersion of a base matrix $\bB$  is given by
\begin{equation}
                       rank(\bH ) = \mu _0 + c_1\mu _1 + . . . + c_{\kappa -1}\mu _{\lambda -1}.
\end{equation}
Note that $\mu _0$ and $\mu _1$ are the ranks of $\bB _0 = \bB ^{\circ 0}$ and $\bB  = \bB ^{\circ 1}$, respectively. If we know $c_0, c_1, . . . , c_{\lambda -1}$ and $\mu _0, \mu _1, . . . , \mu _{\lambda -1}$, then we can compute the rank of $\bH$  from (24).  This can be done by first partitioning the set $\{\bB ^{\circ 0}, \bB , \bB ^{\circ 2}, . . . , \bB ^{\circ (e-1)} \}$ into $\lambda$  conjugate classes $\Psi _0, \Psi _1, . . . , \Psi_{ \lambda -1}$ and determining the rank of the conjugate matrices in each class. Then, use (24) to compute the rank of $\bH $.  This can be carried out systematically.

   As pointed above, the powers of matrices in a conjugate class $\Psi _i = \{\bB ^{\circ t_i}, \bB ^{\circ 2t_i}, . . . , \bB ^{\circ 2^{c_i-1} t_i}\}$ form a cyclotomic coset ${\bf Z}_i = \{t_i, 2t_i, . . . , 2^{c_i-1}t_i\}$ of 2 modulo $e$.   Therefore, to find the conjugate class $\Psi _i$ is equivalent to find the powers of the matrices in $\Psi _i$.  In the following, we present a recursive construction of the cyclotomic cosets, ${\bf Z}_0, {\bf Z}_1, . . . , {\bf Z}_{\lambda -1}$ of 2 modulo $e$. In the construction, the first elememt $t_i $ of each cyclotomic coset ${\bf Z}_i$ is always the smallest integer, the coset representative. In this case, the matrix $\bB^{\circ t_i}$ is the representative of the conjugate class $\Psi _i$.  We begin with the cyclotomic class $\Psi _0 = \{0\}$ which contains only the integer 0.  Suppose we have completed the construction of the $i$th cyclotomic coset $\Psi _i$ (i.e., $\Psi _0, \Psi _1, . . . , \Psi _{i-1}$ have been constructed).  To construct the $i$th  cyclotomic coset $\Psi _i$, we choose the smallest integer $t_i$ in the set ${\cal E} = \{0, 1, . . . , e -1\}$ but not in any of cyclotomic cosets ${\bf Z}_0, {\bf Z}_1, . . . , {\bf Z}_{i-1}$. With $t_i $ as the first element of the cyclotomic coset ${\bf Z}_i$, we form ${\bf Z}_i = \{t_i, 2t_i, . . . , 2^{c_i -1}t_i \}$. It is clear that $t_i$ is the smallest integer in ${\bf Z}_i$ and hence the representative of ${\bf Z}_i$. Continue the above construction process until we form all the cyclotomic cosets of 2 modulo $e$.  It is clear that $ t_i  - 1$ must be in one of the cyclotomic cosets, ${\bf Z}_0, {\bf Z}_1, . . . , {\bf Z}_i$.

   In the following, we give an upper bound on the rank of $\bH $.  First, we need the following theorem which was proved in [32].

\begin{theorem}
 Let $\bM_1$ and $\bM_2$ be two $n\times n$ matrices over GF($q$) with ranks $rank(\bM_1)$ and $rank(\bM_2)$, respectively, and $\bM_1\circ \bM_2$ be the Hadamard product of $\bM_1 $ and $\bM_2$.  Then the rank of $\bM_1\circ \bM_2$, denoted by $rank(\bM_1\circ \bM_2)$, satisfies the following inequality:
\end{theorem}
\begin{equation}
                       rank(\bM_1\circ \bM_2) \leq  rank(\bM_1) \times rank(\bM_2).    
                       \end{equation}
 For the matrices  $\bM_1$ and $\bM_2$ which are not square matrices, we construct square matrices  $\hat{\bM}_1$ and $\hat{\bM}_2$ by adding extra zero rows or columns. Clearly, adding or deleting zero rows or columns does not affect rank. Then, we have 
 \[
 \begin{array}{ll}
 rank(\bM_1 \circ \bM_2)& = rank(\hat{\bM}_1 \circ \hat{\bM}_2)  \\
   & \leq rank(\hat{\bM}_1)\times rank(\hat{\bM}_2)= rank(\bM_1)\times rank(\bM_2).
 \end{array}
 \]

   Consider the class $\Psi _i = \{\bB ^{\circ t_i}, \bB ^{\circ 2t_i}, . . . , \bB ^{\circ 2^{c_i -1} t_i}  \}$ of conjugate matrices. It follows from our construction of the cyclotomic cosets that $t_i$ is the smallest integer in cyclotomic coset ${\bf Z}_i$. Hence, $\bB ^{\circ t_i}$ is the representative matrix of the conjugate class $\Psi_i$. If $t_i - 1$ is contained in the ${i^*}$th cyclotomic coset ${\bf Z}_{{i^*}}$ with ${i^*} < i$, then
\begin{equation}
                                  \bB ^{\circ t_i} = \bB \circ \bB ^{\circ (t_i-1)}.  
\end{equation}
Since $t_i -1$ is an integer in the cyclotomic coset ${\bf Z}_{{i^*}}$, $\bB ^{\circ (t_i-1)}$ is a member matrix in the conjugate class $\Psi _{{i^*}}$ with representative $\bB ^{\circ t_{i^*}}$ where $t_{i^*}$ is the representative in the cyclotomic coset ${\bf Z}_{i^*}$.  Since $i^* < i$ and $t_{i^*} \leq  t_i -1$, then we must have $t_{i^*} < t_i$.  Since the rank of the matrices in conjugate class $\Psi _{i^*}$  is $\mu _{i^*}$, then $rank(\bB ^{\circ (t_i-1)}) = \mu_{ i^*}$. It follows from Theorem 3 that the rank of $\bB ^{\circ t_i}$  is upper bounded by $\mu _1\times \mu _{i^*}$, i.e.,
\begin{equation}
                                 rank(\bB ^{\circ t_i} ) \leq  \mu _1\times \mu _{i^*}.                
\end{equation}
Since $\bB ^{\circ t_i}$ is an $m\times n$ matrix, the rank of $\bB ^{\circ t_i}$ must be upper bounded by $ \min\{m,n\}$, i.e.,
\begin{equation}
                                       rank(\bB ^{\circ t_i} ) \leq  \min\{m,n\}.
\end{equation}
It follows from (24), (27) and (28) that we have the following  upper bound on  $rank(\bB ^{\circ t_i} )$:
\begin{equation}
                                 rank(\bB ^{\circ t_i} ) \leq  \min\{m, n, \mu _1\times \mu _{i^*} \}.
\end{equation}
It follows from (24) and (29) that we obtain an upper bound on the rank of the parity-check matrix $\bH$  which is given by the next theorem.

\begin{theorem}
Let $\bH$  be an $m\times n$ array of CPM's and ZM's over GF(2) of size $e\times e$ given by the $e$-fold dispersion of an $m\times n$ matrix over GF($q$) with $q = 2^r$ and $e = 2^r - 1$.  Let $  \mu _0$ and $\mu _1$ be the ranks of $\bB _0 = \bB ^{\circ 0}$ and $\bB $, respectively. Then, the following gives an upper bound on the rank of $\bH $,
 \begin{equation}
                               rank(\bH ) \leq  \mu _0 +  \sum\limits^{\kappa-1}_{i=1} c_i \times \min\{m, n, \mu _1\mu _{i^*} \}.    
\end{equation}
\end{theorem}
   To construct the cyclotomic cosets systematically, we fill a table with $\kappa$  rows and five columns shown in Table 1. The first column gives the indices of the rows. The entries in the $i$-th row of the third column are the integers in the cyclotomic coset ${\bf Z}_i$ with its representative put in the first position.  The first row of the third column gives the cyclotomic coset  $\Psi _0 = \{0\}$.  The $i$th entry of the fourth column of the table gives the upper bound $\mu _1\mu _{i^*}$ on the rank $\mu _i$ of matrices in the conjugate class $ \Psi _i$ whose Hadamard powers are integers in the cyclotomic coset ${\bf Z}_i = \{t_i, 2t_i, . . . , 2^{(c_i-1)} t_i\}$ given in the $i$th row of the third column. The entry in the $i$th column of the fifth column is the true rank $\mu _i$ of the matrices in $i$th conjugate class $\Psi _i$. The $i$th entry in the second column of the table is the row index $i^*$ for which the cyclotomic coset ${\bf Z}_{i^*}$ contains the integer $t_i  - 1$.  The first entry of the second column is set to 0.

\begin{table}
	\centering
		\begin{tabular}{c|c|c|c|c}
		\hline
			 Row index  &  Row index  &   Cyclotomic coset    &   Upper bound  &  True rank\\
			 \hline
			  $i$         &           $i^* $&            ${\bf Z}_i$ &         $\mu_1\times \mu_{i^*}$  &    $\mu_i$\\
			  \hline
		    0         &        0          &                0               &     $\min\{m,n\}$&  \\
		    \hline
		    1         &         0         &                                &             $\mu_1$           &     $ \mu_1$\\
		    \hline
		     & & \vdots& &\\
		     \hline
		     $\kappa-1$ &  &  & &\\
		      \hline	 
		\end{tabular}
	\caption{Recursive construction of cyclotomic cosets}
	\label{tab:Re}
\end{table}                       
    Once the table is formed, we have all the information of the number of integers in each cyclotomic coset and the rank of each Hadamard power of the base matrix $\bB $.  From these information and using (24), we can compute the rank of the parity-check matrix $\bH $.

   The upper bound on the rank of $\bH$  is very tight as will be shown by examples given in later sections.  If the base matrix $\bB$  has special structures, a combinatorial expression for $rank(\bH )$ can be derived.

   Next we derive an upper bound on $rank(\bH )$ which only depends on the size and the ranks of $\bB _0$ and $\bB $.  Consider the $t$-th Hadamard power $\bB ^{\circ t}$ of the base matrix $\bB$  for $0 < t < e$.  Let $\tau (t)$ denote the number 1's in  the binary representation of $t$ and $f_0, f_1 . . . , f_{ \tau (t)}$ denote the position of these 1's.  Then
\begin{equation}
                                t =   \sum\limits^{ \tau (t) - 1}_{i=0}  2^{f _i }
 \end{equation}
The number $\tau (t)$ is called the weight of the integer $t$. Since $0 < t < e = 2^r - 1$, we have $1 \leq  \tau (t) < r$. Using the above binary representation of $t$, the $t$-th Hadamard power $\bB ^{\circ t}$ of the base matrix $\bB$  can be expressed as the following Hadamard product of $\bB ^{\circ 2^{f _0}}, \bB ^{\circ 2^{f _1}}, . . . , \bB ^{\circ 2^{f_{ \tau (t) - 1}}}$:
\begin{equation}
                           \bB ^{\circ t} =  \prod\limits^{\tau (t) - 1}_{i=0} {\circ}\bB ^{\circ 2^{f _i}}                                                      \end{equation}
Since each term in the product of (32) is a conjugate matrix of the base matrix $\bB $, they are all in the same conjugate class with $\bB$  as the representative and hence they have the same rank $\mu _1$ as $\bB $.  It follows from Theorem 3 and (32) that the rank of $\bB ^{\circ t}$ is upper bounded as follows:
\begin{equation}
                          rank(\bB ^{\circ t}) \leq  \mu _1^{\tau (t)}.
\end{equation}
  Since for $0 < i < r$, there are ${r\choose i}$ nonzero integers less than $e = 2^r - 1$ with weight $i$.   Then it follows from (19) and (33) that we derive the following theorem that gives an upper bound on the rank of an array of CPMs and/or ZMs of size $e\times e$.

\begin{theorem}
Let $\bH$  be an $m\times n$ array of CPM's and ZM's over GF(2) of size $e\times e$ given by the $e$-fold dispersion of an $m\times n$ matrix over GF($q$) with $q = 2^r$ and $e = 2^r - 1$.  Let $\mu _0$ and $\mu _1$ be the ranks of $\bB _0 = \bB ^{\circ 0}$ and $\bB $, respectively. Then,  the rank of $\bH$  is upper bounded as follows:
 \begin{equation}\label{tightbound}
                               rank(\bH ) \leq  \mu_0 + \sum\limits^{r-1}_{i=1} {r\choose i}\min\{m, n, \mu _1^i \}.              
                                   \end{equation}
\end{theorem}
The upper bound given by (\ref{tightbound}) depends only on the choice of the base matrix $\bB $.  For several class of $2\times 2$ SM-constrained base matrices, this bound is very tight.  This will shown in latter sections on code construction.

%

\subsection{Row Redundancy }

   A very  \emph{important structure} of the geometrically and algebraically constructed cyclic or QC-LDPC codes is that their parity-check matrices have large  \emph{row redundancies}, i.e., the parity-check matrix of a code has a large number of dependent rows.  Extensive simulation results show that large row redundancy and large minimum distance make the iterative decoding of an LDPC code to converge at a very fast rate and provide a very low error-floor. Here we would like to show the impact of row-redundancy by an simple example of the $(255,175)$ EG-LDPC code given in \cite{KLF01}. From Figure 1, it is clear that the performance of the SPA decoding algorithm  of the EG-LDPC code improves as the row-redundancy of $\bH$ increases. The performance of $\bH$ with all row-redundancy 0.6863 is about 0.6 dB better than the performance of $\bH$ with no  row-redundancy 0 at \emph{bit error rate} (BER) $10^{-5}$.  
  
\begin{figure}
 \centering
\includegraphics[width = 3in, height=2in]{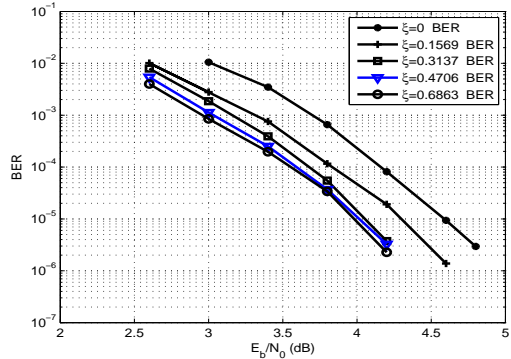}
\vspace{-0.2cm}
\caption{The performance of the $(255,175)$ EG-LDPC code given in \cite{KLF01}.}
\vspace{-0.5cm}
\label{example6}
\end{figure}

   It is interesting that such regular algebraic LDPC codes are generally better than regular LDPC codes designed by pseudo-random matrices [16], [19]. However, irregular algebraic LDPC codes perform close to irregular pseudo-random LDPC codes. We call two Tanner graphs \emph{similar} if they have the same girth and degree-distributions. It is well known that the same message-passing algorithm [12] performs closely under two similar Tanner graphs at the same \emph{signal power per code bit to noise ratio} (SNR$_{cb}$). Without loss of generality, suppose that the parity-check matrix of the first code has redundancy and the second does not. Thus, the null space of the first one defines a higher-rate code than the null space of the second one. Therefore, the first code requires lower \emph{signal power per information bit to noise ratio} (SNR) than the second code to achieve the same performance. In another word, the redundancy will improve the performance of LDPC codes. Most regular algebraic LDPC codes have redundancies, but most irregular algebraic LDPC codes have no redundancies. Thus, we see this interesting phenomenon from simulations.

\begin{definition}
Let $\bH$  be an $m\times n$ matrix over GF(2). Let $rank(\bH )$ denote the rank of $\bH $.  The row redundancy $\bH$  is defined as the ratio $\xi  = (m - rank(\bH ))/m$.  If $\bH$  has full rank, i.e., $rank(\bH ) = m$, then its row redundancy is zero. 
\end{definition}

It follows from the above definition that $m - rank(\bH ) $ is simply the number of redundant rows (or dependent rows) of $\bH $.

   The parity-check matrix $\bH$  obtained by $e$-fold array dispersion of a $2\times 2$ SM-constrained base matrix $\bB$  is, in general, not full rank and in fact, has a large number of redundant rows. Using the upper bound on the rank of the parity-check matrix $\bH$  given by (34), we obtain the following lower bound on the redundant rows of $\bH $, denoted $R(\bH )$:
\begin{equation}                     
                                R(\bH ) \leq  me - \mu _0 - \sum\limits^{r-1}_{i=1}{r \choose i} \times \min\{m, n, \mu_1^{i}\}.
\end{equation}
  RC-constrained parity-check matrices of QC-LDPC codes with large row redundancies have been reported in [20]-[25]. These parity-check matrices are all constructed by array dispersions of $2\times 2$ SM-constrained base matrices using finite fields.  They are arrays of CPMs and/or ZMs.  In the next section, several classes of these base matrices will be briefly described.

Most high-rate QC-LDPC codes in applications have small $m$, since code rate is lower bounded by $1 - m/n$. Thus, $\mu_1^{\tau(t)}$ is larger than $\min(m, n)$ with a small $\tau(t)$. Thus, we consider all $\bB^{\circ t} $'s have full rank, except for $t=0$ and $t=1, 2, 2^2,\cdots, 2^{c_1-1}$. Recall that $c_t$ is the smallest positive integer satisfying $t 2^{c_t}\equiv t \mbox{mod} e$.  Since $e=2^r-1$, $c_1=r$. In this case, $\bB_0$ and $\bB^{\circ \tau(t)}$ with small $\tau(t)$ are the key factors to design parity-check matrices with redundancies. Without loss of generality, suppose $m\leq n$. Then, we give the next two Corollaries.

\begin{corollary} \label{binary_rank}
If the parity-check matrix $\bH$ consisting of circulants and zero matrices of size $e$ is binary and the rank of its base matrix $rank(\bB)=\mu_1$, then the redundant rows of $\bH$ is at least $r (m-\mu_1)$. 
\end{corollary}
 \beginofproof
 Considering all $\bB^{\circ t} $'s have full row-rank, except for $t=1, 2, 2^2,\cdots, 2^{r-1}$, from (35), we have 
\[
\begin{array}{lll}
R(\bH)&\geq& me - \mu_0 - \sum\limits^{r-1}_{i=1} {r \choose i} \min\{m, n,  \mu_1^{i}\}, \\
      &\geq& (m-\mu_1)r.
\end{array}
\]

\endofproof

\begin{corollary} \label{binary_rank_4}
If the parity-check matrix $\bH$ consisting of circulants and zero matrices of size $e=2^r-1$ is binary and its base matrix does not have zero entry, then its rank $rank(\bH)$ is at most $m(e-1)-c_1 (m-r_1)+1$. 
\end{corollary}
 \beginofproof
If all entries of $\bB$ are nonzero, then $\bB_0$ is all 1's. Thus, $\mu_0=1$. Then $R(\bH)\geq r(m-\mu_1)+m-1$.
\endofproof

The above two Corollaries give us two guidelines to design parity-check matrices with redundancies. \emph{First, the base matrix $\bB$ should have rank as small as possible. Second, it should contain few zero entries.}

Similarly, from (19), we can give a bound for the non-binary QC-LDPC codes whose parity-check matrix is an array of $\alpha$-multiplied CPMs and zero matrices. 
\begin{corollary} \label{nonbinary_rank}
If the parity-check matrix $\bH_{\alpha}$ consists of an array of $\alpha$-multiplied CPMs and zero matrices of size $e=2^r-1$, then there are at least $r(m-\mu_1)$ redundant rows. 
\end{corollary}

\section{Rank Enumerations of Two Well Known Classes of RC-Constrained Low-Density  Arrays of CPMs}
 
   In  this section, we consider two well known classes of $2\times 2$ SM-constrained base matrices.  From these two classes of base matrices, two classes of RC-constrained arrays of CPMs and/or ZMs can be constructed.  The null spaces of the arrays in these two classes give two classes of RC-constrained QC-LDPC codes. 
 
\subsection{Latin Squares}
 
   A Latin square of order $n$ is an $q\times q$ array for which each row and each column contains every element of a set of $n$ elements exactly once [29].  For any given field GF($q$), there is a $q\times q$ Latin square whose entries are elements of the field.  In a recent paper [23], it was proved that a Latin square over the field GF($q$) satisfies the $2\times 2$ SM-constrained and hence can be used as a base matrix for constructing a $q\times q$ RC-constrained array $\bH_{lat}$ of CPMs and ZMs of size $(q - 1)\times (q - 1)$.  The null space of any subarray of $\bH_{lat}$ gives an RC-constrained QC-LDPC code.  Using the recursive algorithm and bounds developed in Section V, we can enumerate the rank of any subarray of $\bH_{lat}$.

   Consider the field GF($q$) where $q=2^r$.  Let $\alpha$ be a primitive element of GF($q$).  Then, $\alpha ^{-\infty}  = 0$, $\alpha ^0 = 1$,  $\alpha , \alpha ^2, . . . , \alpha ^{q-2}$ give all the elements of GF($q$).  The following $q\times q$ matrix over GF($q$) gives a Latin square of order $q$:
\begin{equation}     
 \bB _{lat} =\left[\begin{array}{ccccc}                             
                              1 - 1    &            1- \alpha        &  . . .&       1 - \alpha ^{q-2}  &        1 - 0\\
                              \alpha  - 1        &        \alpha  - \alpha      &   . . .  &     \alpha  - \alpha ^{q-2} &             \alpha  - 0\\
                                  \vdots &          & \ddots&& \vdots\\
                         \alpha ^{q-2} - 1 &    \alpha ^{q-2} - \alpha &  . . .&   \alpha ^{q-2} - \alpha ^{q-2}&      \alpha ^{q-2} - 0\\
                               0 -1   &             0 - \alpha       &   . . .  &    0 - \alpha ^{q-2}&              0 - 0
\end{array}\right].
\end{equation}
The entries on the main diagonal of $\bB_{ lat}$ are the 0-element of GF($q$).  This matrix satisfies the $2\times 2$ SM-constrained and hence can be used as a base matrix for code construction.

For $1 \leq  m, n \leq  q$, let $\bB_{lat}(m,n)$ be an $m \times n$ submatrix of $\bB_{lat}$, taken from the upper-left corner of $\bB_{lat}$. Let $\bH_{lat}(m,n)$ be an $m \times n$ subarray of $\bH_{lat}$, obtained by the $e$-fold dispersion of $\bB_{lat}(m,n)$.  $\bH_{lat}(m,n)$ is an $me \times ne$ matrix over GF(2). Then, the null space of $\bH_{lat}(m,n)$ gives a QC-LDPC code of length $ne$.  If $\bH_{lat}(m,n)$ does not contain any ZM of $\bH_{lat}$, it has constant column and row weights $m$ and $n$, respectively.  Then the QC-LDPC code given by the null space of $\bH_{lat}(m,n)$ is an ($m$,$n$)-regular QC-LDPC code.  If $\bH_{lat}(m,n)$ constains ZMs of $\bH_{lat}$, then $\bH_{lat}(m,n)$ has two different column weights, $m-1$ and $m$. In this case, the null space of $\bH_{lat}(m,n)$ gives a near-regular QC-LDPC codes.

In [9], a combinatorial expression of the rank of $\bH_{lat}(m,n)$ have been given for the case of $n=q$. In this section, we will generalize this combinatorial expression to make it suitable for the case of either $m \geq \frac{q}{2}$ or $n \geq \frac{q}{2}$. We will also show that in this case, the equality of the upper bound specified in (\ref{tightbound}) holds. Since $\bB_{lat}$ is a symmetric matrix, we only need to give the proof for the case of $n \geq \frac{q}{2}$.
 
Define the index sequence ${\cal A}=(0,1,2,...,q-2,-\infty)$, and let ${\cal A}_m$, ${\cal A}_n$ denote two index sets consisting of the first $m$ and $n$ components of ${\cal A}$, respectively. Then for $1 \leq t <2^r-1$ we have $\bB_{lat}^{\circ t}(m,n)=[(\alpha^i+\alpha^j)^t]_{i \in {\cal A}_m, j \in {\cal A}_n}$. Since the characteristic of GF($q$) is 2, in the binomial expansion of $(\alpha^i+\alpha^j)^t$, only the terms with odd coefficients exist.  Let $\theta_t$ be the number of odd coefficients in the binomial expansion of $(\alpha^i+\alpha^j)^t$ (or the number of odd integers in the $t$-th level of \emph {Pascal triangle}). Let $l_1, l_2, . . . , l_{\theta_t}$ denote the positions of these odd coefficients.  We note that $l_1 = 0$ and $l_{\theta_t} = t$. It is clear that $\theta_t \leq t + 1$. Then 
\begin{equation}\label{binomial1}
(\alpha^i+\alpha^j)^t=\alpha^{it}+\alpha^{i(t-l_2)} \alpha^{jl_2}+\alpha^{i(t-l_3)}\alpha^{jl_3}+...+\alpha^{i(t-l_{\theta_{t-1}})} \alpha^{jl_{\theta_{t-1}}}+\alpha^{jt}
\end{equation}
Let $\gamma_m$ and $\gamma_n$ denote the $m$-th and $n$-th component in the index sequence $\cal A$ repectively. Based on the expression given by (\ref{binomial1}), the $t$-th Hadamard power $\bB_{lat}^{\circ t}(m,n)$ can be expressed as a product of two matrices as follows:
\begin{equation}\label{e44}
\bB_{lat}^{\circ t}(m,n)={\bf V}_{t,L}(m,n){\bf V}_{t,R}(m,n)
\end{equation}
with 
\begin{eqnarray}\label{e66}
{\bf V}_{t,L}(m,n)=\left[\begin{array}{ccccc}
\alpha^0 & \alpha^0 &\alpha^0 & \cdots & 1 \\
\alpha^t & \alpha^{t-l_2}&\alpha^{t-l_3}& \cdots & 1 \\
\alpha^{2t} & \alpha^{2(t-l_2)} &\alpha^{2(t-l_3)}& \cdots & 1 \\
\vdots    & \vdots  & \vdots  & \ddots & \vdots      \\
\alpha^{\gamma_mt}& \alpha^{\gamma_m(t-l_2)}&\alpha^{\gamma_m(t-l_3)}& \cdots & 1 \\
\end{array} \right],\nonumber\\
{\bf V}_{t,R}(m,n)=\left[\begin{array}{ccccc}
1& 1&1& \cdots & 1 \\
\alpha^0 & \alpha^{l_2}&\alpha^{2l_2}& \cdots & \alpha^{\gamma_nl_2} \\
\alpha^0 & \alpha^{l_3}&\alpha^{2l_3}& \cdots & \alpha^{\gamma_nl_3} \\
\vdots    & \vdots  & \vdots  & \ddots & \vdots      \\
\alpha^0 & \alpha^{l_{\theta_t}}&\alpha^{2l_{\theta_t}}& \cdots & \alpha^{\gamma_nl_{\theta_t}} \\
\end{array} \right],
\end{eqnarray}
where ${\bf V}_{t,L}$ is a $m \times \theta_t$ matrix over GF$(2^r)$ and ${\bf V}_{t,R}$ is a $\theta_t \times n$ matrix over GF$(2^r)$.

Let $\omega(t)$ be the number of nonzero terms in the radix-2 expansion(or binary representation) of $\theta_t$, called the radix-2 weight of $\theta_t$.  It follows form Lucas theorem [29] that $\theta_t = 2^{\omega(t)}$. For $0\leq t<2^r-1$, since $\theta(t) \leq t + 1 < 2^r$, we must have $\omega(t)<r$ and $\theta_t \leq 2^{r-1}=\frac{q}{2}$. Since $n \geq \frac{q}{2}$, we conclude that $\theta_t \leq n$.

Based on the structure of the two matrices ${\bf V}_{t,L}(m,n)$ and ${\bf V}_{t,R}(m,n)$, we readily know that both of them can be transformed into Vandermonde Structure by elementary row and column operation. In this case we know that ${\bf V}_{t,R}(m,n)$ has full row rank, i.e. $rank({\bf V}_{t,R}(m,n))=\theta_t$, and  $rank({\bf V}_{t,L}(m,n))=\min\{m,\theta_t\}=\min\{m,2^{\omega(t)}\}$. Thus we have
\begin{equation}\label{e144}
rank(\bB_{lat}^{\circ t}(m,n))=rank({\bf V}_{t,L}(m,n){\bf V}_{t,R}(m,n))=rank({\bf V}_{t,L}(m,n))=\min\{m,2^{\omega(t)}\}
\end{equation}

Then based on (\ref{brank}) it follows that

\begin{eqnarray}\label{e145}
&&rank(\bH_{lat}(m,n))=rank(\bB_{lat}^{\circ 0}(m,n))+\sum_{t=1}^{2^r-2} rank(\bB_{lat}^{\circ t}(m,n)) \nonumber\\
&&=rank(\bB_{lat}^{\circ 0}(m,n))+\sum_{t=1}^{2^r-2} \min\{m,2^{\omega(t)}\} \nonumber\\
&&=\mu_0(m,n)+\sum^{r-1}_{i=1} {r \choose i}\min\{m,2^{i}\}
\end{eqnarray}
where $\mu_0(m,n)=rank(\bB_{lat}^{\circ 0}(m,n)$. On the other hand, it follows from (\ref{tightbound}) that 

$$rank(\bH_{lat}(m,n)) \leq \mu_0(m,n) + \sum\limits^{r-1}_{i=1} {r\choose i}\min\{m, n, (\mu_1(m,n))^i \}$$
where $\mu_1(m,n)=rank(\bB_{lat}(m,n))$. It follows from (\ref{e144}) that $rank(\bB_{lat}(m,n)=min(m,2)=2$ (We never choose $m=1$ for code construction). Thus (\ref{tightbound}) becomes
\begin{equation}\label{e146}
rank(\bH_{lat}(m,n)) \leq \mu_0(m,n) + \sum\limits^{r-1}_{i=1} {r\choose i}\min\{m, n, 2^i \}
\end{equation}

Since $\min\{m,n,2^i\} \leq \min\{n,2^i\}$, comparing (\ref{e145}) and (\ref{e146}) we see that the equality in (\ref{e146}) must hold. Let $\omega_m$ be the largest integer such that $2^{\omega_m} \leq m$, then another combinatorial expression of $rank(\bH_{lat}(m,n))$  can be derived based on (\ref{e145}):

\begin{eqnarray}\label{combinexpress}
&&rank(\bH_{lat}(m,n))=\mu_0(m,n)+\sum^{r-1}_{i=1} {r \choose i}\min\{m,2^{i}\} \nonumber\\
&&=\mu_0(m,n)+\sum_{{1 \leq i \leq r-1} \atop {2^i \leq m}}2^i+\sum_{{1 \leq i \leq r-1} \atop {2^i > m}}m \nonumber\\
&&=\mu_0(m,n)+\sum_{i=1}^{\min\{\omega_m,r-1\}} 2^i + \sum_{i=\omega_m+1}^{r-1}m
\end{eqnarray}
where the second sum term exists only if $\omega_m+1 \leq r-1$. The rank expression in (\ref{combinexpress}) is suitable for the case of $n \geq \frac{q}{2}$. For the case of $m \geq \frac{q}{2}$, based on (\ref{combinexpress}) and the fact that $\bB_{lat}$ is a symmetric matrix we can directly obtain the combinatorial expression of $rank(\bH_{lat}(m,n))$ as follows:

\begin{eqnarray}\label{combinexpress1}
rank(\bH_{lat}(m,n))=\mu_0(m,n)+\sum_{i=1}^{\min\{\omega_n,r-1\}} 2^i + \sum_{i=\omega_n+1}^{r-1}n
\end{eqnarray} 
where $\omega_n$ is the largest integer such that $2^{\omega_n} \leq n$, and the second term exists only if $\omega_n+1 \leq r-1$.

Using (\ref{combinexpress}) and (\ref{combinexpress1}), we can yield some more interesting result with respect to some special case. Firstly we consider the case that $m=n \geq \frac{q}{2}$. In this case we have  $\min\{\omega_m,r-1\}=\min\{\omega_{m+1},r-1\}=r-1$, and the second sum term of (\ref{combinexpress1}) doesn't exists. Thus we have $$rank(\bH_{lat}(m+1,m+1))-rank(\bH_{lat}(m,m))=\mu_0(m+1,m+1))-\mu_0(m,m)).$$ Since it has been known in \cite{HLG_12} that if $k_0$ denote the number of 0's in $\bB_{lat}^{\circ 0}(m,n)$, we have 
\begin{equation} \label{rankb0}
\mu_0(m,n)=\left\{\begin{array}{ll}
k_0+1 & \mbox{ if $k_0<\min\{m,n\}$,}\\
k_0-1 & \mbox{ if $k_0=n=m$ is odd,}\\
k_0   & \mbox{ otherwise.}\end{array}\right.
\end{equation}

Based on (\ref{rankb0}), we readily know that for even $m$, $\mu_0(m+1,m+1)=\mu_0(m,m)$, and for odd $m$, $\mu_0(m+1,m+1)=\mu_0(m,m)+2.$ Thus we have the following recursive relationship in the case of $m \geq \frac{q}{2}$:

\begin{equation} \label{Eq_rankB0}
rank(\bH_{lat}(m+1,m+1))=\left\{\begin{array}{ll}
rank(\bH_{lat}(m,m)) & \mbox{ if $m$ is even,}\\
rank(\bH_{lat}(m,m))+2 & \mbox{ if $m$ is odd,}\end{array}\right.
\end{equation}

Secondly in the case of $n \leq \frac{q}{2}$, we consider two subarrays $\bH_{lat}(q,n)$ and $\bH_{lat}(\frac{q}{2},n)$. In this case, we have  $\min\{\omega_q,r-1\}=\min\{\omega_\frac{q}{2},r-1\}=r-1$, and the second sum term of (\ref{combinexpress1}) doesn't exists. Hence based on (\ref{combinexpress1}) we have $$rank(\bH_{lat}(q,n))-rank(\bH_{lat}(\frac{q}{2},n))=\mu_0(q,n)-\mu_0(\frac{q}{2},n).$$ Since $n \leq \frac{q}{2}$, based on (\ref{rankb0}) we conclude that $\mu_0(q,n)-\mu_0(\frac{q}{2},n)$. Hence in case of $n \leq \frac{q}{2}$ we have 
\begin{equation} \label{latinredu1}
rank(\bH_{lat}(q,n))=rank(\bH_{lat}(\frac{q}{2},n))
\end{equation}
equation (\ref{latinredu1}) indicates that in the case of $n \leq \frac{q}{2}$, all the last $\frac{(q-1)q}{2}$ rows of  $\bH_{lat}(q,n)$ are redundant rows, hence the null spaces of $\bH_{lat}(q,n)$ and $\bH_{lat}(\frac{q}{2},n)$ give the same QC-LDPC codes.

Thirdly in the case of $n>\frac{q}{2}$, we consider two subarrays $\bH_{lat}(q,n)$ and $\bH_{lat}(n,n)$. In this case we have $\min\{\omega_q,r-1\}=\min\{\omega_n,r-1\}=r-1$, and the second sum term of (\ref{combinexpress1}) doesn't exists. Hence based on (\ref{combinexpress1}) we have $$rank(\bH_{lat}(q,n))-rank(\bH_{lat}(n,n))=\mu_0(q,n)-\mu_0(n,n).$$ It follows from (\ref{rankb0}) that $\mu_0(q,n)=n$, and that $\mu_0(n,n)=n-1$ for odd $n$, and $\mu_0(n,n)=n$ for even $n$. Hence in case of $n \leq \frac{q}{2}$ we have 

\begin{equation} \label{latinredu2}
rank(\bH_{lat}(q,n))=\left\{\begin{array}{ll}
rank(\bH_{lat}(n,n)) & \mbox{ if $n$ is even,}\\
rank(\bH_{lat}(n,n))+1 & \mbox{ if $n$ is odd,}\end{array}\right.
\end{equation}
equation (\ref{latinredu2}) indicates that in the case of even $n>\frac{q}{2}$, all the last $(q-1)(q-n)$ rows of  $\bH_{lat}(q,n)$ are redundant rows, hence the null spaces of $\bH_{lat}(q,n)$ and $\bH_{lat}(n,n)$ give the same QC-LDPC codes.

Suppose we take $m=n=q$, i.e. we use the whole matrix $\bH_{lat}$ as the parity-check matrix, $\bH_{lat} $ has $2^r(2^r - 1) - 3^r + 1 = 4^r - 3^r - 2^r + 1$ redundant rows, which satisfies the equality in (35).  For $r \geq  4$, the number of redundant rows is very large.  The null space of $\bH_{lat}$  gives an RC-constrained $(4^r-2^r, 4^r -3^r - 2^r +1) $ QC-LDPC code with minimum distance at least $2^r + 1$, whose Tanner graph has a girth of at least 6.

\begin{example}   Consider the $64\times 64$ Latin square constructed based on GF($2^6$) using the form of (36).  It is a $64\times 64$ matrix over GF($2^6$).  We use this matrix as the base matrix $\bB _{lat}$. Hence $m = n = 64$ and $e = 63$.  The 63-fold array dispersion of $\bB_{ lat}$ gives a $64\times 64$ array $\bH_{lat}$  of CPMs and ZMs  of size $63\times 63$ with the ZM's lying on the main diagonal of  $\bH _{lat}$.   $\bH_{lat}$  is a $4032\times 4032$ matrix with both column and row weights 63.  Hence the code has minimum distance at least 64. The rank of $\bB _{lat}$ is 2 and the rank of $\bB _0$ is 64. Table 2 gives the cyclotomic cosets of 2 modulo 63 and the ranks of the conjugate matrices of $\bB _{lat}$ and their bounds.  There are 13 cyclotomic cosets of 2 modulo 63 and hence the 63 Hadamard powers of the base matrix $\bB _{lat}$  are grouped into 13 conjugate classes.  From Table 2, we see that for $0 \leq  i < 13$, the upper bound $\mu _1\mu _{i^*}$ is equal to the true ranks of matrices in the conjugate class $\Psi _i$.  It follows from (24) that $rank(\bH_{lat}) = 728$ which is exactly equal to $3^6 - 1$.  This shows that the bound given by (30) (or the bound given by (34)) is tight. The parity-check matrix $ \bH_{lat}$ has 3304 redundant rows and a redundancy $\xi  = 0.8194$.  The performances of this code over AWGN channel decoded with 5, 10 and 50 iterations of the SPA are shown in Figure 2.  We see that decoding of this code converges very fast. At the BER of $10^{-6}$, the performance gap between 5 and 50 iterations is about 0.1 dB.  The performance curves for 10 and 50 iterations almost overlap.  Also at the BLER (block error rate) of $10^{-5}$, the code decoded with 10 iterations of the SPA performs 1.2 dB from the sphere packing found. 
\end{example}

\begin{table}
	\centering
		\begin{tabular}{c|c|c|l|c|c}
		\hline
		   Row index  &  Row index  & Weight &      Cyclotomic coset       & Upper bound  & True rank\\
     \hline
       $ i $        &         $i^*$   &   $\tau(t)$       &             $\qquad\quad {\cal  Z}_i$                                  &      $u_1^{\tau(t)}$           &        $u_i$\\
       \hline                                                       
        0         &         N/A          &  N/A &   0                                     &        N/A           &     N/A\\
    \hline 
        1         &         1      &  1     &  1     2    4     8    16   32   &         2              &       2      \\ 

\hline
                           2        &        2   & 2          &  3     6   12   24   48   33   &         4              &       4\\
    \hline 
                              3        &        2  & 2          &  5    10  20   40   17   34   &         4              &       4\\
    \hline 
                          4         &          3     & 3      &  7    14  28   56   49   35   &         8              &       8\\
    \hline 
                         5         &          2    & 2        &  9    18   36                       &         4              &       4\\
    \hline 
                          6         &          2  & 3        & 11   22   44  25   50  37    &         8              &       8   \\
    \hline 
                          7         &          3     & 3      & 13   26   52  41   19  38    &         8              &       8  \\
    \hline 
                          8         &          5     &4      & 15   30   60  57   51  39    &        16             &      16\\
    \hline 
                          9         &          4       & 3    & 21   42                              &         8              &       8\\
    \hline 
                         10        &          7         & 4  & 23   46   29  58   53  43    &        16             &      16\\
    \hline 
                         11        &          8       & 4    & 27   54   45                       &        16             &      16\\
      \hline 
                            12        &          9      &5     & 31   62   61  59   55  47    &        32             &       32   \\
    \hline    
		\end{tabular}
	\caption{Cyclotomic cosets of 2 modulo 63 and the ranks of the conjugate matrices of the base matrix $\bB_{lat}$.}	
	\label{tab:Cyclotomic}
\end{table}
\begin{example}
 Choose $m = 6$ and $n = 64$.  Choose the first 6 rows of $64\times 64$ matrix over GF($2^6$) constructed in Example 1 as the base matrix, denoted $\bB _{lat}(6,64)$, for code construction. $\bB _{lat}(6,64)$ contains 6 zeros, one in each of the first 6 columns.  Array dispersion of this base matrix gives a $6\times 64$ array $\bH_{lat} (6,64)$ of CPM's and ZM's of size $63\times 63$. This array is a $378\times 4032$ matrix over GF(2) with row weight 63 and two column weights 5 (315 columns) and 6 (3717 columns).  Using the upper bound on $rank(\bB _{lat}(6,64)^{\circ t_i })$ given by (34), we find that the rank of $\bH_{lat} (6,64)$ is upper bounded by 324.  However, the actual rank of $\bH_{cpm,lat} (6,64)$ is  also 324.  Therefore, the bound gives the actual rank of $\bH_{lat} (6,64)$. The matrix $\bH_{cpm,lat} (6,64)$ has 54 redundant rows. The null space of the array $\bH_{lat} (6,64)$ gives a (4032,4708) QC-LDPC code with rate 0.92.  The error performances of this code over the AWGN channel decoded with 5, 10 and 50 iterations are shown in Figure 3.
\end{example}
 
\begin{example}
 Consider the $6\times 64$ array constructed in Example 2.  Each of the first 6 column blocks contains a single ZM.  If we remove the first 6 column blocks of the array, we obtain a $6\times 58$ subarray $\bH_{lat}(6,58)$ of the entire $64\times 64$ array $\bH_{lat}$ constructed in Example 1.  The subarray $\bH_{lat}(6,58)$ is a $378\times 3654$ matrix over GF(2) with column and row weights 6 and 58, respectively.   Using the upper bound on $rank(\bB _{lat}(6,58))$ given by (34), we find that the rank of $\bH_{lat} (6,58)$ is upper bounded by 324.
\end{example}

\subsection{Vandermonde Matrices}
 
    A special type of Vandermonde matrices also satisfies the $2\times 2$ SM-constraint.   Again, we consider the field GF($q$) with $q = 2^r$.  Let $p$ be the largest prime factor of $q - 1$.  Let $k$ be the integer such that $q - 1 = kp$.  If $q - 1$ is prime, then $p = q - 1$ and $k = 1$. Let $\alpha$ be a primitive element of GF($q$) and $\beta =\alpha^k$.  Then, the order of $\beta$ is $p$.  Form the following $p\times p$ matrix over GF($q$):
\begin{equation}
         \bB_{van}  =  \left[ \begin{array}{ccccc}
                                1        &      1   &      \cdots   &   1    &         1 \\
                                    1         &       \beta     &    \cdots  &     \beta^{p-2}    &        \beta^{p-1} \\
                                    \vdots &                     &                    \ddots  &   & \vdots\\
                           1        &     \beta^{p-1}   &    \cdots&  \beta^{(p-1)\times (p-2)}    &         \beta^{(p-1)\times (p-1)}  
         \end{array}\right].
\end{equation} 
It can be readily proved that $\bB_{van}$ satisfies the $2\times 2$ SM-constraint and hence can be used for constructing RC-constrained QC-LDPC codes.  The rank of $\bB_{van}$ is $p$.  In fact, all Hadamard powers of $\bB_{van}$, $\bB_{van} ^{\circ 1} = \bB_{van},  \bB_{van}^{\circ 2}, \ldots, \bB_{van}^{\circ (p-1)}$, have the same rank $p$ since $p$ is a prime. Since all the entries of $\bB_{van}$ are nonzero elements of GF($q$), all the entries of $\bB_{van}^{\circ 0}$ are 1's.  Consequently, $rank(\bB_{van}^{\circ 0}) = 1$.
 
   The CPM array dispersion of $\bB_{van}$ results in a $p\times p$ array $\bH_{van}$ of CPMs of size $e\times e$ with $e = 2^r - 1$. It is a $pe\times pe$ matrix over GF(2) with both column and row weights equal to $p$.  It follows from (19) that the rank of $\bH_{van}$ is
 \begin{equation}
                            rank(\bH_{van}) = 1 + (p - 1)p.
                             \end{equation}
The row redundancy is then $\xi (\bH_{van})= (2^r - p - 1)/(2^r - 1)$.  In case $q - 1 = 2^r - 1$ is a prime, then $rank(\bH_{van}) = 1 + (2^r - 2)(2^r - 1)$.

   For $1 \leq  m, n \leq  p $ and $m\leq   n$, let $\bB_{van}(m,n)$ be an $m\times n$ submatrix of $\bB_{van}$ which is still a Vandermonde matrix with rank $m$.  The ranks of $\bB_{van}^{\circ 1}(m,n) = \bB_{van}(m,n)$,  $\bB_{van}^{\circ 2}(m,n)$, $...,$  $\bB_{van}^{\circ (p-1)}(m,n)$ are $ m$.  The rank of $\bB_{van}^{\circ 0}(m,n)$ is 1.  The CPM array dispersion of $\bB_{van}(m,n)$ results in an $m\times n$ array $\bH_{van}(m,n)$ of CPMs of size $e\times e$.  It follows from (19), the rank of $\bH_{van}(m,n)$
 \begin{equation}
                          rank(\bH_{van}(m,n)) = mp - m + 1. 
 \end{equation}
The null space of $\bH_{van}(m,n)$ gives an RC-constrained ($m$,$n$)-regular QC-LDPC code of length $n(2^r - 1)$.  If $n = p$, the base matrix $\bB_{van}(m,p)$ is actually the parity-check matrix of a non-primitive RS code over GF($q$).

\section{Construction of QC-LDPC Codes Based on Random Partition of Finite Fields and its Rank Analysis}
In this section, we present a new algebraic method for constructing a class of QC-LDPC codes. Given a finite field, we first partition the elements of the field into two disjoint subsets (any partition).  Based on these two disjoint subsets, we form a matrix over the given field.  Every entry of the matrix is a sum of two elements, one from one subset and the other from the second subset. From this matrix, we can form an array of CPMs.  This array, as a matrix, satisfies the RC-constraint.  Then, the null space of this array gives a QC-LDPC code.
 
\subsection{A Class of $2 \times 2$ SM-Constrained Base Matrices Constructed by Field Partitions}
Let GF$(2^r)$ be a finite field with $2^r$ elements which is an extension field of the binary field GF$(2)$. Let $\alpha$ be a primitive element of GF$(2^r)$.  Then, the powers of $\alpha$, $\alpha^{-\infty}=0, \alpha^0 = 1, \alpha, \alpha^2,...,\alpha^{2^r- 2}$, give all the elements of  GF$(2^r)$ and $\alpha^{2^r -1} = 1$.
Let $m$ and $n$ be two positive integers such that $m+ n = 2^r$.  Partition the elements of GF$(2^r)$ into two disjoint subsets, $G_1=\{\lambda_0,\lambda_1,...,\lambda_{m-1}\}$ and $G_2=\{\delta_0,\delta_1,...,\delta_{n-1}\}$, i.e., $G_1 \cup G_2=$GF$(2^r)$ and $G_1 \cap G_2= \phi$.  Form the following $m \times n$ matrix over GF$(2^r)$:
\begin{equation}\
{\bB_{rp}}=\left[\begin{array}{cccc}
\lambda_0+\delta_0 & \lambda_0+\delta_1 & \cdots & \lambda_0+\delta_{n-1} \\
\lambda_1+\delta_0 & \lambda_1+\delta_1 & \cdots & \lambda_1+\delta_{n-1} \\
\vdots    & \vdots    & \ddots & \vdots      \\
\lambda_{m-1}+\delta_0 & \lambda_{m-1}+\delta_1 & \cdots & \lambda_{m-1}+\delta_{n-1} \\
\end{array} \right]
\end{equation}
We note that each row of $\bB_{rp}$ is formed by adding an element in $G_1$ to all the $n$ elements in $G_2$ and each column of $\bB_{rp}$ is formed by adding an element in $G_2$ to all the $m$ elements in $G_1$.  Since the characteristic of the field GF$(2^r)$ is 2, every element in GF$(2^r)$ is its own additive inverse.  It follows from the fact that $G_1$ and $G_2$ are disjoint and the formation of $\bB_{rp}$, we readily prove that: (1) all the entries in $\bB_{rp}$ are nonzero; (2) all the entries in a row of $\bB_{rp}$ are different; and (3) all the entries in a column of $\bB_{rp}$ are different.  Every entry in $\bB_{rp}$ is a power of the primitive element $\alpha$ of GF$(2^r)$.
In the following, we prove that the $m \times n$ matrix $\bB_{rp}$ over GF$(2^r)$ given by (1) satisfies the $2 \times 2$ SM-constraint. Hence, $\bB_{rp}$ can be used as the base matrix to construct an $m \times n$ RC-constrained binary array of CPMs. 
\begin{theorem} 
The $m \times n$ matrix $\bB_{rp}$ over GF$(2^r)$ given by (1) satisfies the $2 \times 2$ SM-constraint.
\end{theorem}
\beginofproof 
Since all the entries of $\bB_{rp}$ are nonzero.  To prove the theorem,we only need to prove that every $2 \times 2$ submatrix is non-singular.  Consider a $2 \times 2$ submatrix of $\bB_{rp}$:
\begin{equation}
{\bf Q}=\left[\begin{array}{cccc}
\lambda_i+\delta_k & \lambda_i+\delta_t\\
\lambda_j+\delta_k & \lambda_j+\delta_t\\ 
\end{array} \right]
\end{equation}
where $0 \leq i,j < m$ and $0 \leq k,t <n$ and $i \neq j$ and $k \neq t$. Since $i \neq j$ and $k \neq t$, then $\lambda_i \neq \lambda_j$ and $\delta_t \neq \delta_k$.  If this matrix is singular, then $$(\lambda_i+\delta_k )(\lambda_j+\delta_t)-(\lambda_i+\delta_t)(\lambda_j+\delta_k)=0.$$
Manipulating the above equality, we have $(\lambda_i - \lambda_j)(\delta_t - \delta_k) = 0$ which implies either $\lambda_i = \lambda_j$ or $\delta_t = \delta_k$.  This contradicts the fact that $\lambda_i \neq \lambda_j$ and $\delta_t \neq \delta_k$.  Therefore, any $2 \times 2$ submatrix $\bf Q$ of $\bB_{rp}$ is non-singular.  This proves the theorem. 
\endofproof

\subsection{Construction of Codes}
Since the base matrix $\bB_{rp}$ over GF$(2^r)$ satisfies the $2 \times 2$ SM-constraint.  It can be used to construct an $m \times n$ RC-constrained binary array $\bH_{rp}$ of CPMs and ZMs. Since each entry in $\bB_{rp}$ is a nonzero element in GF$(2^r)$, it must be a power of $\alpha$ which is a primitive element of GF$(2^r)$. Label the rows  and columns of a $(2^r - 1) \times (2^r - 1)$ CPM from 0 to $2^r - 2$ which correspond to powers of $\alpha$, $\alpha^0 = 1,\alpha,\alpha^2,...,\alpha^{2^r - 2}$. For $0 \leq i< m$ and $0 \leq j < n$, let $\lambda_i + \delta_j = \alpha^{k_{i,j}}$ with $0 \leq k_{i,j}< 2^r - 1$. Then the construction of $\bH_{rp}$ directly from $\bB_{rp}$ is carried out as follows: replacing the entry $\lambda_i + \delta_j = \alpha^{k_{i,j}}$ at the $i$th row and $j$th column of $\bB_{rp}$ by a $(2^r - 1) \times (2^r - 1)$ CPM whose top row (called the generator) has a single 1-component at the position $k_{i,j}$.  This gives the array $\bH_{rp}$ corresponding to the base matrix $\bB_{rp}$.  Since all the entries are nonzero, $\bH_{rp}$ contains no ZM and is an array of CPMs only. 

$\bH_{rp}$ is a $m(2^r - 1) \times n(2^r - 1)$ matrix over GF$(2)$.  Since each CPM has both column and row weights equal to 1, the column and row weights of $\bH_{rp}$, as a $m(2^r- 1) \times n(2^r - 1)$ matrix over GF$(2)$, are $m$ and $n$ respectively.  Consequently, the null space of $\bH_{rp}$ gives an RC-constrained QC-LDPC code of length $n(2^r-1)$ with minimum distance at least $m+1$ whose Tanner graph has a girth at least 6.  Note that the null space of any sub-array of $\bH_{rp}$ also gives a QC-LDPC code.

\subsection{Rank Analysis of the Parity-Check Matrices}
Next we will analyze the rank to show that the resulting parity-check matrix is rich in redundant rows and satisfies the two guidelines derived from Corollary 1 and 2. First, because there are no zero elements in $\bB_{rp}$, $\bB_{0}$ is an all `1' matrix and has rank $\mu_0=1$. Second, if we rewrite $\bB_{rp}$ as the product of two matrix,
\[
\bB_{rp} = \bV_{L} \bV_{R},
\]
where 
\[
\bV_{L}= \left[  
\begin{array}{ll}
\lambda_0 & 1\\
\lambda_1 & 1\\
\lambda_2 & 1\\
\vdots\\
\lambda_{m -1} & 1\\
\end{array}
 \right] \qquad \mbox{ and } \qquad \bV_{L}= \left[  
\begin{array}{lllll}
1 & 1 & 1 & \ldots & 1\\
\zeta_0 & \zeta_1 & \zeta_2 & \ldots & \zeta_{n-1}\\
\end{array}
 \right]
\]
are full row-rank and full column-rank, respectively, it is clear that the rank of $\bB_{rp}$ is small $\mu_1=\min\{m, n, 2\}$. From Theorem 5, an upper bound on the rank of $\bH_{rp}$ can be given from the base matrix $\bB_{rp}$,
\begin{eqnarray}\label{equality}
 rank(\bH_{rp}) &\leq&  \mu_0 + \sum\limits^{r  - 1}_{i=1} {r \choose i} \min\{m,n,\mu_1^{i} \},  \nonumber\\
 & \leq & 1+ \sum\limits^{r  - 1}_{i=1} {r \choose i}   \min\{m,n,2^{i} \}.
\end{eqnarray}
In the above section, we show that this bound is very tight and in some cases the equality holds. Here we will show that for parity-check matrices constructed based on random partition, the equality holds again in some special cases.  Here we only give the proof for the case $m\leq n$. If $m> n$, the proof is similar.

For $0 \leq t < 2^r - 1$, we have $\bB_{rp}^{\circ t} = [(\lambda_i + \delta_j)^t]$, $0 \leq i< m$ and $0 \leq  j < n$.  In the binomial expansion of $(\lambda_i + \delta_j)^t$, only the terms with odd coefficients exist since the odd coefficients modulo-2 are equal to 1 while even coefficients modulo-2 become zeros.  Let $\theta_t$ be the number of odd coefficients in the binomial expansion of $(\lambda_i + \delta_j)^t$ (or the number of odd integers in the $t$-th level of \emph {Pascal triangle}). Let $l_1, l_2, . . . , l_{\theta_t}$ denote the positions of these odd coefficients.  We note that $l_1 = 0$ and $l_{\theta_t} = t$. It is clear that $\theta_t \leq t + 1$. Then 
\begin{equation}\label{binomial}
(\lambda_i+\delta_j)^t=\lambda_i^t+\lambda_i^{t-l_2} \delta_j^{l_2}+\lambda_i^{t-l_3} \delta_j^{l_3}+...+\lambda_i^{t-l_{\theta_{t-1}}} \delta_j^{t_{\theta_{t-1}}}+\delta_j^t
\end{equation}
Based on the expression given by (\ref{binomial}), the $t$-th Hadamard power $\bB_{rp}^{\circ t}$ can be expressed as a product of two matrices as follows:
\begin{equation}\label{e4}
\bB_{rp}^{\circ t}={\bf V}_{t,L}{\bf V}_{t,R}
\end{equation}
with 
\begin{eqnarray}\label{e6}
{\bf V}_{t,L}=\left[\begin{array}{ccccc}
\lambda_0^t& \lambda_0^{t-l_2}&\lambda_0^{t-l_3}& \cdots & 1 \\
\lambda_1^t& \lambda_1^{t-l_2}&\lambda_1^{t-l_3}& \cdots & 1 \\
\lambda_2^t& \lambda_2^{t-l_2}&\lambda_2^{t-l_3}& \cdots & 1 \\
\vdots    & \vdots  & \vdots  & \ddots & \vdots      \\
\lambda_{m-1}^t& \lambda_{m-1}^{t-l_2}&\lambda_{m-1}^{t-l_3}& \cdots & 1 \\
\end{array} \right],\nonumber\\
{\bf V}_{t,R}=\left[\begin{array}{ccccc}
1& 1&1& \cdots & 1 \\
\delta_0^{l_2}& \delta_1^{l_2}&\delta_2^{l_2}& \cdots & \delta_{n-1}^{l_2} \\
\delta_0^{l_3}& \delta_1^{l_3}&\delta_2^{l_3}& \cdots & \delta_{n-1}^{l_3} \\
\vdots    & \vdots  & \vdots  & \ddots & \vdots      \\
\delta_0^{l_{\theta_t}}& \delta_1^{l_{\theta_t}}&\delta_2^{l_{\theta_t}}& \cdots & \delta_{n-1}^{l_{\theta_t}} \\
\end{array} \right],
\end{eqnarray}
where ${\bf V}_{t,L}$ is a $m \times \theta_t$ matrix over GF$(2^r)$ and ${\bf V}_{t,R}$ is a $\theta_t \times n$ matrix over GF$(2^r)$.
Let $\omega(t)$ be the number of nonzero terms in the radix-2 expansion(or binary representation) of $\theta_t$, called the radix-2 weight of $\theta_t$.  It follows form Lucas theorem [29] that $\theta_t = 2^{\omega(t)}$. For $0\leq t<2^r-1$, since $\theta(t) \leq t + 1 < 2^r$, we must have $\omega(t)<r$ and $\theta_t<2^{r-1}$.

To determine the rank of $\bB_{rp}^{\circ t}$, we need to determine the ranks of ${\bf V}_{t,L}$ and ${\bf V}_{t,R}$. This can be done with the following case: the nonzero elements of both $G_1$ and $G_2$ form two sequences of consecutive powers of $\alpha$.  The 0 element can be in either $G_1$ or $G_2$.  For example, $G_1 = \{0, \alpha^0, \alpha, . . . , \alpha^{m-2}\}$ and $G_2 = \{\alpha^{m-1}, \alpha^m, . . . , \alpha^{2^r - 2}\}$.   Let $n = 2^r - m$. We also assume that $m \leq 2^{r-1} \leq n$.  In this case, ${\bf V}_{t,L}$  and ${\bf V}_{t,R}$  can be transformed into matrices with the Vandermonde structure [27],[28] by elementary column and row operations.  Since $m \leq n$, then $n \geq 2^{r-1} > \theta_t$.   As a result, $rank({\bf V}_{t,L}) = min(m,\theta_t)$ and $rank({\bf V}_{t,R}) = \theta_t$.  It follows from (\ref{e4}) that $rank(\bB_{rp}^ {\circ t}) = rank({\bf V}_{t,L} {\bf V}_{t,R})$.  Since ${\bf V}_{t,R}$ has full row rank,
\begin{equation}
rank(\bB_{rp}^{\circ t})=rank({\bf V}_{t,L}{\bf V}_{t,R})=rank({\bf V}_{t,L})=min(m,\theta_t)
\end{equation}
for $0 \leq t <2^r-1$. Since $rank(\bB_{rp}^ {\circ 0}) = \mu_0=1$. Then, it follows from (19) that the rank of $\bH_{rp}$ is:

\begin{eqnarray*}\label{equalityhold}
&&rank(\bH_{rp})=1+\sum_{t=1}^{2^r-2} rank(\bB_{rp}^{\circ t}) \nonumber\\
&&=1+\sum_{t=1}^{2^r-2} min(m,\theta_t) \nonumber\\
&&=1+ \sum^{r}_{i=1} {r \choose i}min(m,2^{i})
\end{eqnarray*}
Therefore, the equality in (\ref{equality}) holds.
Let $\omega_0$ be the largest integer such that $2^{\omega_0} \leq m$ ,Then another combinatorial expression for the sum terms given by (\ref{equalityhold}) can be derived as follows:
{\small
\begin{eqnarray*}\label{e9} 
&&\sum_{t=1}^{2^r-2}rank(\bB_{rp}^{\circ t})=\sum_{t=1}^{2^r-2}min(m,\theta_t) \nonumber\\
&&=\sum_{{1 \leq t \leq 2^r-2} \atop {\omega_0<\omega(t)}}m+\sum_{{1 \leq t \leq 2^r-2} \atop {\omega_0 \geq \omega(t)}} \theta_t \nonumber\\ 
&&=\sum_{\omega=\omega_0+1}^{m-1} \nonumber \sum_{\omega(t)=\omega}m+\sum_{\omega=1}^{\omega_0}  \sum_{\omega(t)=\omega}2^{\omega(t)} \nonumber\\
&&=\sum_{\omega=\omega_0+1}^{r-1} {r \choose \omega}m +\sum_{\omega=1}^{\omega_0} {r \choose \omega}2^\omega \nonumber\\
&&=\sum_{\omega=\omega_0+1}^{r-1} {r \choose \omega}m +\sum_{\omega=1}^{r-1} {r \choose \omega}2^\omega-\sum_{\omega=\omega_0+1}^{r-1} {r \choose \omega}2^\omega. \nonumber\\ 
\end{eqnarray*}
}
Note that $\sum_{\omega=1}^{r-1} {r \choose \omega}2^\omega=3^r-2^r-1$. Consequently, we have 
\begin{equation}
\sum_{t=1}^{2^r-2}rank(\bB_{rp}^{\circ t})=3^r-2^r-1-\sum_{\omega=\omega_0+1}^{r-1} {r \choose \omega}(2^\omega-m). 
\end{equation}
Since $rank(\bB_{rp}^{\circ 0})=1$, we have the following combinatorial expression for the rank of $\bH_{rp}$ with $m \leq 2^{r-1}$:
\begin{equation}\label{e10}
rank(\bH_{rp})=3^r-2^r-\sum_{\omega=\omega_0+1}^{r-1} {r \choose \omega}(2^\omega-m).
\end{equation}

For $m = 2^{r-1} = n$, $\bB_{rp}$ is a square matrix, denoted by $\bB_{rp,s}$ and its corresponding array $\bH_{rp,s}$ is a $2^{r-1} \times 2^{r-1}$ array of $(2^r - 1) \times (2^r - 1)$ CPMs.  In this case, $\omega_0 = r - 1$ and
\begin{equation}\label{e12}
rank(\bH_{rp,s})=3^r-2^r.
\end{equation}
The null space of $\bH_{rp,s}$ gives a binary QC-LDPC code with the following parameters: (1) length $n = 2^{r-1}(2^r - 1)$; (2) dimension $k = 2^{2r-1}-3^r$; (3) minimum distance $d_{min}$ is at least $2^{r-1}+ 1$.  The null space of any sub-array of $\bH_{rp,s}$ gives a QC-LDPC code.

\begin{example}
Let GF$(2^6)$ as the field for construction.  Let $\alpha$ be a primitive element of GF$(2^6)$.  Partition GF$(2^6)$ into two subsets, $G_1 = \{0, 1, \alpha, \alpha^2, \alpha^3, \alpha^4\}$ and $G_2 = \{\alpha^5, \alpha^6, . . . , \alpha^{62}\}$.  Using these two subsets of GF$(2^6)$, we can construct a $6 \times 58$ base matrix $\bB_{rp}$ over GF$(2^6)$.  Array dispersion of this base matrix results in a $6 \times 58$ array $\bH_{rp}$ of $63 \times 63$ CPMs.  $\bH_{rp}$ is a $378 \times 3654$ RC-constrained matrix over GF$(2)$ with column and row weights 6 and 58, respectively.  Since $m = 6$, we find that $\omega_0 = 2$. Using the combinatorial expression given by (\ref{e10}), we find that the rank of $\bH_{rp}$ is 319.  Hence, the null space of $\bH_{rp}$ gives a (6,58)-regular (3654,3335) QC-LDPC code with rate 0.9126.  The bit and block error performances of the code decoded with 5, 10 and 50 iterations of the MSA are shown in Figure 2. We see that the decoding of the code converge fast.  At the BER of $10^{-6}$, the code decoded with 50 iterations of the MSA perform 1.2 dB from the Shannon limit.  At the BLER (block error rate) $10^{-5}$, it performs 0.8 dB from the sphere packing bound. 
\end{example}
\begin{figure}[!t]
\centering
\includegraphics[width=3.5in]{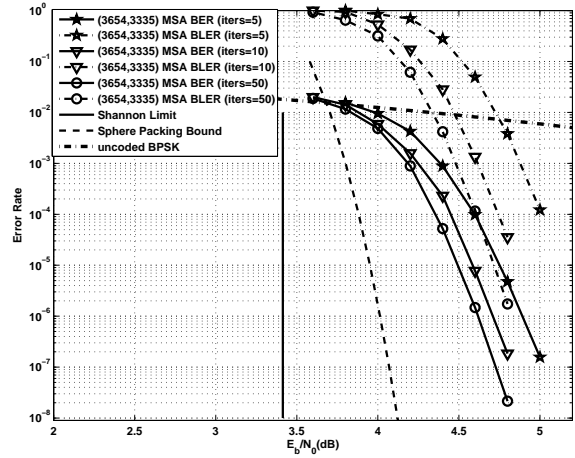}
\caption{ The Error Performance of QC-LDPC code given in Example 4} \label{delta}
\end{figure}

Similarly, by replacing each entry of $\bB_{rp}$ by $\alpha$-multiplied CPM's, we can construct non-binary RC-constrained arrays $\bH_{\alpha,rp}$ over $GF(2^r)$ based on field partitions.  The null spaces of these arrays give a class of non-binary QC-LDPC codes.

\begin{example}
Let GF$(2^5)$ be the field for code construction and $\alpha$ be a primitive element of the field.  Partition the elements of GF$(2^5)$ into two disjoint subsets, $G_1 = \{0, 1, \alpha,\alpha^2\}$ and $G_2 = \{\alpha^3, \alpha^4,...,\alpha^{30}\}$.  Based on these two subsets of GF$(2^5)$, we form a $4 \times 28$ base matrix $\bB_{rp}$ over GF$(2^5)$ of the form given by (40).  Replacing each entry in $\bB_{rp}$ by its corresponding $\alpha$-multiplied CPM of size $31 \times 31$, we obtain a $4 \times 28$ array $\bH_{\alpha,rp}$ of $\alpha$-multiplied CPMs of size $31 \times 31$.  $\bH_{\alpha,rp}$ is a $124 \times 868$ RC-constrained matrix over GF$(2^5)$ with column and row weights 4 and 28, respectively.  Since $m = 4$, the parameter $w_0$ is 2. Using the combinatorial expression given by (48), we find that the rank of $\bH_{\alpha,rp}$ is 111.  The null space of $\bH_{\alpha,rp}$ gives a (4,28)-regular 32-ary (868,757) QC-LDPC code with rate 0.8722.  The bit, symbol and block error performances of this code decoded with 50 iterations of fast Fourier transform $q$-ary SPA (FFT-QSPA) are shown in Figure 3.  At the BLER of $10^{-5}$, the code performs 1.72 dB from the sphere packing bound. 
\end{example}

\begin{figure}[!t]
\centering
\includegraphics[width=3.5in]{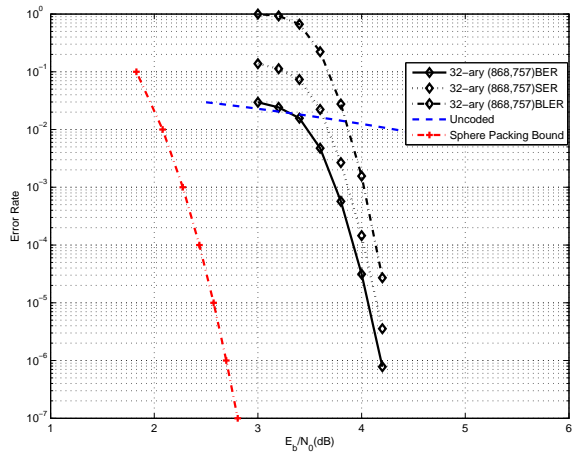}
\caption{ The Error Performance of QC-LDPC code given in Example 5.} \label{delta}
\end{figure}

\section{New Constructions of QC-LDPC Codes }

\subsection{QC-LDPC Codes by Diamond-Shape Dispersion}
The construction of QC-LDPC codes, \emph{diamond-shape} dispersion, is first mentioned in [16], which have good performance and are well-known for correcting single burst erasures. In this subsection, we will put it into a more general form such that its degree distributions and dispersion size are more flexible and its rank is possible for analysis. In addition, we will propose a construction method which leads to more row redundancy from the two guidelines.

Let $\bW$ be an $m_w\times n_w$ base matrix over GF($q$) with rank $\mu_w$, where $q=2^r$, $n_w\geq 2$ and $m_w$ is a factor of $n_w=c_w m_w$. We divide it into two parts: the upper matrix $\bW_{u}$, where 
\begin{equation}
w_{u,i,j}= \left\{
\begin{array}{ll}
w_{i,j},  & \mbox{if } i \leq j/c_w \\
0, & \mbox{else}
\end{array}\right.
\end{equation}
and the lower matrix $\bW_{d}$, where 
\begin{equation}
w_{l,i,j}= \left\{
\begin{array}{ll}
w_{i,j},  & \mbox{if } i> j/c_w \\
0, & \mbox{else}
\end{array}\right.
\end{equation}
such that $\bW=\bW_{u}+\bW_{l}$. Then we form a new $m\times n$ base matrix 
\begin{equation}\label{b4}
\bB_{ds}=\left[\begin{array}{l:l}
 \bW_{u} &  \bW_{l}    \\
\hdashline
 \bW_{l} &  \bW_{u}      
\end{array}
\right], 
\end{equation}
where $m=2m_w$ and $n=2n_w$. The parity-check matrix $\bH_{cpm,ds}$ is the $e$-fold dispersion of the base matrix $\bB_{ds}$, where $e=2^r-1$. Since 
$$rank(\bB_{ds})=rank\left(\left[\begin{array}{c:c}
 \bW_{u}+\bW_{l}  &  \bW_{l} +\bW_{u}    \\
\hdashline
 \bW_{l} &  \bW_{u}      
\end{array}
\right] \right)=rank\left(\left[\begin{array}{c:c}
\bW  &  \bW   \\
\hdashline
 \bW_{l} &  \bW_{u}      
\end{array}
\right] \right),$$ 
the rank of $\bB_{ds}$ is at most $\mu_w+m_w$. Considering its $t$-th Hadamard power,
\begin{equation}
rank(\bB_{ds}^{\circ t})=rank\left(\left[\begin{array}{c:c}
 \bW_{u}^{\circ t}+\bW_{l}^{\circ t}  &  \bW_{l}^{\circ t} +\bW_{u}^{\circ t}     \\
\hdashline
 \bW_{l}^{\circ t} &  \bW_{u}^{\circ t}      
\end{array}
\right] \right)=rank\left(\left[\begin{array}{c:c}
 \bW^{\circ t}  &  \bW^{\circ t}     \\
\hdashline
 \bW_{l}^{\circ t} &  \bW_{u}^{\circ t}      
\end{array}
\right] \right),
\end{equation} 
From (33), we have the rank of $\bB_{ds}^{\circ t}$ is at most $\mu_w^\tau(t)+ m_w$. Thus, assuming $m\leq n$, from Theorem 5, 
\begin{equation}
rank(\bH_{cpm,ds}) \leq  \mu_0 + \sum\limits^{r-1}_{i=1} {r\choose i}\min\{ 2m_w,  \mu _w^i +m_w \}.              
\end{equation}
Thus, there are at least 
\begin{equation}
R(\bH_{cpm,ds}) \geq  (2m_w-\mu_0) + \sum\limits^{r-1}_{i=1} {r\choose i}\min\{ 0, m_w-   \mu _w^i  \}            
\end{equation}
redundant rows. Similar to Corollary 3, suppose that $\bW$ contains all nonzero entries, i.e., $\mu_0=m_w+1$, such that $\bH_{cpm,ds}$ have more redundant rows, 
\begin{equation}
R(\bH_{cpm,ds}) \geq  (m_w-1) + \sum\limits^{r-1}_{i=1} {r\choose i}\min\{ 0, m_w-   \mu _w^i \}.           
\end{equation}
From the above redundancy analysis, it is straightforward to construct $\bH_{cpm,ds}$ with redundancies. The base matrix $\bW$ should be low rank (the base matrices of random partition LDPC codes, Latin square LDPC codes and EG LDPC codes with rank 2 satisfy this requirement). Moreover, from (52), $\bW$ should not have zero entries. 

\begin{example}
Based on the Latin Square over GF($2^5$), we construct a base matrix $6\times 24$ $\bW$ with rank 2, which does not contain zero entries. Then, we form a new base matrix $\bB_{ds}$ from (50). From $31$-fold matrix dispersion of $\bB_{ds}$, we obtain a $372\times 1488$ parity-check matrix  $\bH_{cpm,ds}$. Its rank is 327, i.e., there are 45  redundant rows which is exactly lower bounded by  45 (54). The null space of $\bH_{cpm,ds}$ defines a (6,24) QC-LDPC codes with code rate 0.78. In Figure 4, compared with the PEG code with the same code length, column weight and code rate, it has 0.3 dB coding gain at BER $10^{-6}$. Moreover, it is quasi-cyclic which is very cost efficient and is capable of correcting single burst erasures of length less than $5\times 31+1=156$. 
\end{example}

\begin{figure}
 \centering
\includegraphics[width = 4in]{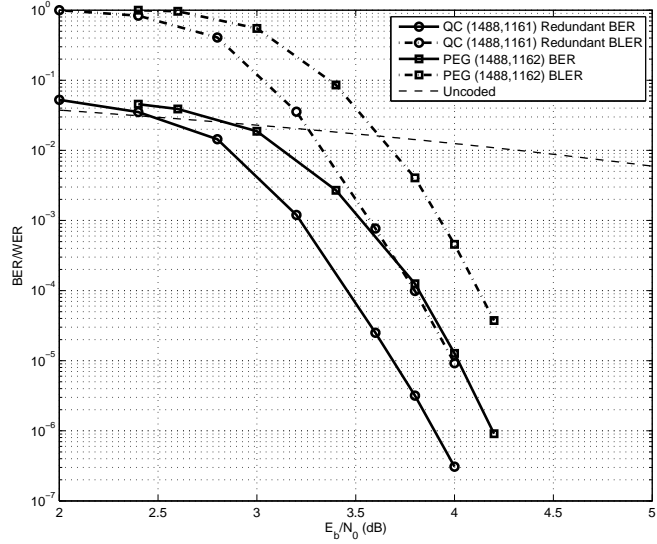}
\caption{The error performances of the $(1488,1161)$ QC-LDPC code given in Example 6.}
\label{example6}
\end{figure}

\subsection{QC-LDPC Codes by Product-like Dispersion}

Let $\left\{ \bW_i,\; 0\leq i< l \right\}$ be a set of $m_w\times n_w$ base matrices over GF($q$). Suppose that the rank of $\bW_i$ is $r_{w,i}$. We form a new base matrix $\bB_{pl}$ from them such that 
\begin{equation}\label{b5}
\bB_{pl}=\left[
\begin{array}{c}
\begin{array}{ccccc}
\bW_0 & {\bf 0} & \cdots & {\bf 0}\\
 {\bf 0}  &  \bW_1 & \cdots &  {\bf 0}\\
 \vdots  & &\vdots &\vdots\\
  {\bf 0} &   {\bf 0}   & \cdots &  \bW_{l}\\
  \hdashline\\
\end{array}\\
{\bf C}_g
\end{array}
\right],
\end{equation}
where ${\bf C}_g$ is an $m_g\times n_w l$ matrix with rank $\mu_g$ such that $\bB_{pl}$ satisfies the $2\times 2$ SM-constraint. The new base matrix $\bB_{pl}$ has dimensions $m=m_g + m_w l $ and $n=n_w l$. Its rank is upper bounded by $\mu_g+\sum \mu_{w,i} $. Similar as (49),  the rank of $\bB_{pl}^{\circ t}$ is upper bounded by $\min\{m_g,  \mu_g^t\} + \sum  \min\{m_w,   \mu_{w,i}^t\}$  . The parity-check matrix $\bH_{cpm,pl}$ is the $e$-fold dispersion of the base matrix $\bB_{pl}$, where $e$ is a factor of $q-1$. Then from Corollary 3, $\bW_i$ should not have zero entries to maximize the redundancies in  $\bH_{cpm,pl}$. Then, the rank of the parity-check matrix $\bH_{cpm,pl}$ is upper bounded by $me-c_1 (m-\sum r_{w,i} + r_g)$.

\begin{example}
Based on the Latin Square over GF($2^4$), we form two matrices $3\times 8$ $\bW_0$ and $\bW_1$. From (\ref{b5}), we have the new base matrix from $\bW_0$, $\bW_1$ and $2 \times 16 $ (2,8) matrix ${\bf C}_g$. From $15$-fold matrix dispersion of $\bB_{pl}$, we obtain a $120\times 240$ parity-check matrix  $\bH_{cpm,pl}$. The null space of $\bH_{cpm,pl}$ defines a (4,8) QC-LDPC codes with code rate 0.56. In Figure 5, compared with the PEG code with the same code length, column weight and code rate, it has 0.2 dB coding gain at BER $10^{-6}$. Moreover, we also constructed a (4,8) PEG code with $120\times 240$   parity-check matrix which is almost full rank for comparison. The proposed code has 0.2 dB coding gain at BER $10^{-6}$.
\end{example}

\begin{figure}
 \centering
\includegraphics[width = 4in]{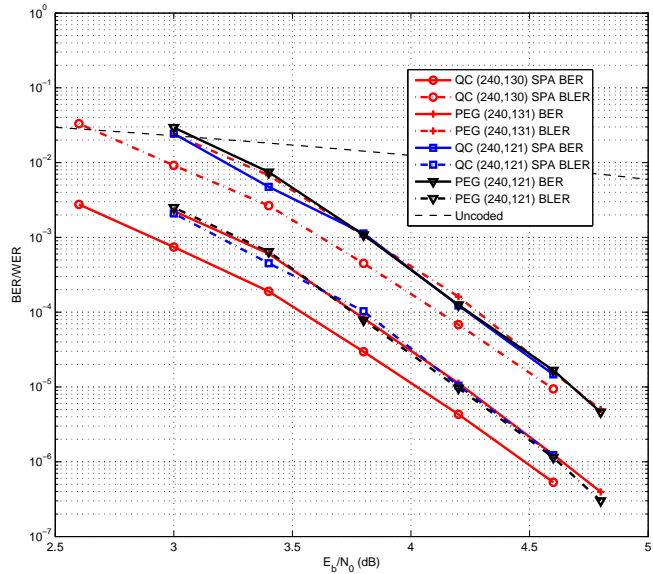}
\caption{The error performances of the $(240,131)$ QC-LDPC code given in Example 7.}
\label{example7}
\end{figure}

\begin{example}
Based on the Latin Square over GF($2^6$), we form three matrices $4\times 18$ $\bW_0$, $\bW_1$ and $\bW_2$. From (\ref{b5}), we have the new base matrix from $\bW_0$, $\bW_1$ and $3 \times 54 $ (1,18) matrix ${\bf C}_g$. From $63$-fold matrix dispersion of $\bB_{pl}$, we obtain a $945 \times 3402$ parity-check matrix  $\bH_{cpm,pl}$. The null space of $\bH_{cpm,pl}$ defines a (5,18) (3402,2502) QC-LDPC codes with code rate 0.74.  In Figure 6, compared with the PEG code with the same code length, column weight and code rate, it has 0.2 dB coding gain at BER $10^{-6}$. 
\end{example}

\begin{figure}
 \centering
\includegraphics[width = 4in]{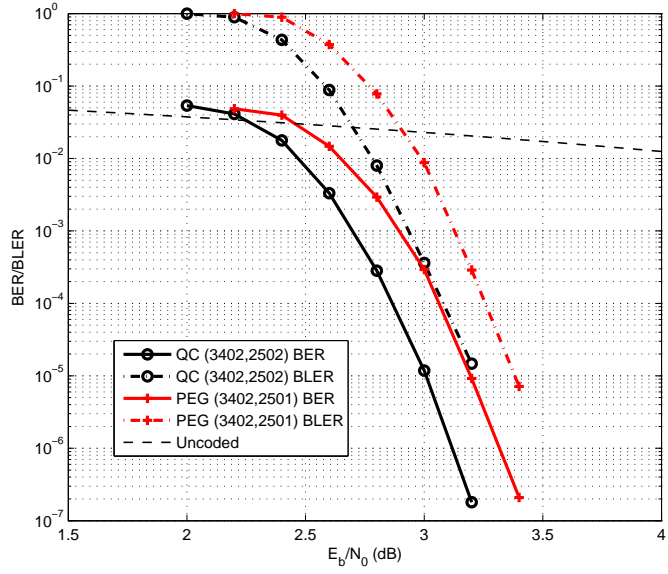}
\caption{The error performances of the $(3402,2502)$ QC-LDPC code given in Example 8.}
\label{example8}
\end{figure}

\subsection{Irregular QC-LDPC Codes by Masking}
It is well known that the parity-check matrices of irregular QC-LDPC codes have few redundant rows. In this sub-section, we design irregular QC-LDPC codes with redundancies for the first time. Let $\bW$ be an $m\times n$ base matrix over GF($q$) with rank $\mu_w$, where $q=2^r$. $\bM$ is the masking matrix with rank $r_m$. Then, a new base matrix $\bB_{ir}$ is formed by the Hadamard product of $\bW$ and $\bM$
$$
\bB_{ir}=\bW \circ \bM.
$$
The parity-check matrix $\bH_{cpm,ir}$ is the $e$-fold dispersion of the base matrix $\bB_{ir}$, where $e=2^r-1$. The rank of $\bH_{cpm,ir}$ is related to $\bB_{ir}$, whose rank is bounded by Theorem 3, $rank(\bB_{ir})\leq rank(\bB_{ir})\times rank(\bM)$. Since there exist base matrices constructed from random partition, Latin squares and Euclidean geometries whose ranks are only 2. The left issue is to find $\bM$ with low rank, so $rank(\bB_{ir})\leq 2 r_m$. If $\bM$ also has good column degree distributions, then we can expect that $\bH_{cpm,ir}$ performs well by message-passing algorithms in the threshold region. Moreover, since $rank(\bH)\leq m e-c_1(m-2r_m)$, extra coding gain from redundancies can be obtained.

Such $\bM$ can be constructed from the circulant matrix $\Phi(\bh)$ formed by the parity-check vector $\bh$ of high-rate cyclic codes ($n$,$k_h$).  $\bM$ can be the sub-matrix or the full-matrix of the product
\begin{equation}
\bC\cdot \Phi(\bh).
\end{equation}
where the $m\times n$ matrix $\bC$ is used to control the degree-distributions  and the dimension of $\bM$. Thus, the rank of $\bM$ is at most $\min\{ n-k_h,m\}$. For example, 
$$
\bC=\left[\begin{array}{ccccc}
1 & 1 & 0 &\cdots&0\\
1 & 0 & 1 &\cdots&0\\
\vdots& & &\vdots&\vdots\\
1 & 0 & 0 & \cdots &1
\end{array}\right].
$$

\begin{example}
Based on the Latin Square over GF($2^6$), we form a $9\times 52$ base matrix $\bW_{ir}$. Since the null space of the minimum polynomial of $x^3+x^2+1$ defines a (63,60) high-rate code, we use its parity-check vector to form a circulant $\Phi(\bh)$ with rank 3. The $\bM$ is a $9\times 52$ sub-array of the product of $\Phi(\bh)$ and $\bC$ such that it has 26 columns with weight 4 and 26 columns with weight 5. The new base matrix $\bB_{ir}$ is formed by the Hadamard product of $\bW$ and $\bM$. From $63$-fold matrix dispersion of $\bB_{ir}$, we obtain a $567 \times 3276$ parity-check matrix  $\bH_{cpm,ir}$ which have 24 redundant rows. The null space of $\bH_{cpm,ir}$ defines a (3276,2733) QC-LDPC codes with code rate 0.834. It only performs 1.4 dB away from the Shannon limit at BER $10^{-6}$ in Figure 7.
\end{example}

\begin{figure}
 \centering
\includegraphics[width = 4in]{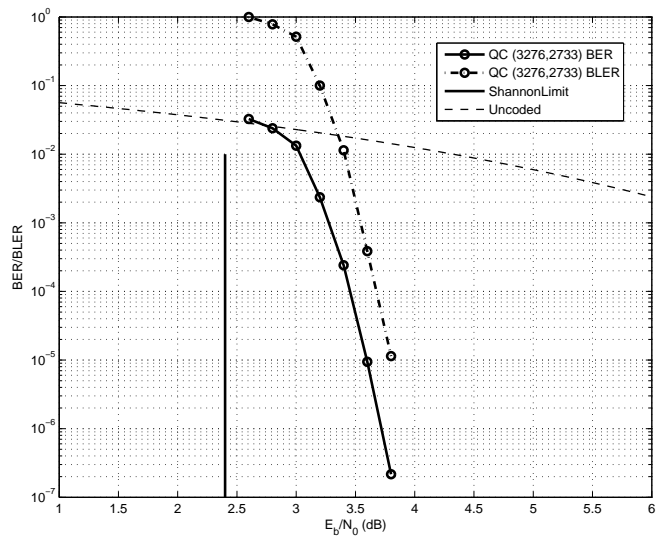}
\caption{The error performances of the $(3276,2733)$ QC-LDPC code given in Example 9.}
\label{example9}
\end{figure}

\section{Conclusion and Remarks}
 Thanks to the rank analyses \cite{K07}, [26] based on the Fourier transform, we could expand the rank analysis and row-redundancy into all QC-LDPC codes.  It is possible to take row-redundancy into account in code constructions such that structured QC-LDPC codes can achieve better performance.  Furthermore, we presented
a class of $2\times 2$ SM-constrained base matrices which are constructed based on partitions of finite fields of characteristic of 2. QC-LDPC codes defined by the null space of these base matrices are flexible for code design. For these codes, the equalities of the above bounds hold.


\begin{thebibliography}{11}

\bibitem{Townsend_67}
R. L. Townsend and E. J. Weldon, Jr., ``Self-orthogonal quasi-cyclic codes,'' \emph{IEEE Trans. Inform. Theory}, vol. IT-13, no. 2, pp. 183-195, Apr. 1967.


\bibitem{Kasami74}
T. Kasami,
``A Gilbert-Varshamov bound for quasi-cycle codes of rate 1/2,''
{\it IEEE Trans. Inform. Theory}, vol. IT-20, no. 5, p. 679, Sep. 1974.


\bibitem{Ga62}
R. G. Gallager,
``Low density parity check codes,''
{\it IRE Trans. Inform. Theory}, vol. IT-8, no. 1, pp. 21--28, Jan. 1962.

\bibitem{MacKay96}
 D. J. C. MacKay and R. M. Neal, ``Near Shannon limit performance of low density parity-check codes,'' \emph{Electro. Lett.}, vol. 32, pp. 1645-1646, Aug. 1996.


\bibitem{M99} D. J. C. MacKay, ``Good error-correcting codes based on very sparse matrices,'' \emph{IEEE
  Trans. Inform. Theory}, vol. 45, no.  2, pp. 399--432, Mar. 1999.

\bibitem{LCZLF06}  Z. Li, L. Chen, L. Zeng, S. Lin and W. Fong, ``Efficient encoding of quasi-cyclic low-density parity-check codes,'' 
\emph{IEEE Trans. Commun.}, vol. 54, no. 1, pp.  71--81, 2006. 


\bibitem{CP04}  Y. Chen and K. Parhi, ``Overlapped message passing for quasi-cyclic
low-density parity check codes,'' \textit{IEEE Trans. Circuits and
Systems I}, vol. 51, no. 6, pp. 1106--1113, Jun. 2004.

\bibitem{WC07} 
Z. Wang and Z. Cui, ``Low-complexity high-speed decoder design for
quasi-cyclic LDPC codes,'' \textit{IEEE Trans. VLSI}, vol. 15, no. 1, pp. 104--114, Jan. 2007.

\bibitem{Tan88}
R. M. Tanner,
``A transform theory for a class of group-invariant codes,''
\emph{IEEE Trans. Inform. Theory}, vol. 34, no. 4, pp. 752--775, Jul. 1988.


\bibitem{KLF01} Y. Kou, S. Lin, and M. P. C. Fossorier, ``Low density parity-check codes based on
  finite geometries: A rediscovery and new results,'' \emph{IEEE Trans. Inform. Theory}, vol. 47,
  no. 7, pp. 2711--2736, Nov. 2001.



\bibitem{LC04} S. Lin and D. J. Costello, Jr., \emph{Error Control Coding: Fundamentals and
    Applications}, 2nd edition. Upper Saddle River, NJ: Prentice Hall, 2004.


\bibitem{RSU01} 
T. Richardson, M. A. Shokrollahi, and R. Urbanke, ``Design of capacity-approaching irregular
 low-density parity-check codes," \textit{IEEE Trans. Inform. Theory}, vol. 47, no. 2, pp. 619--637, Feb. 2001.



\bibitem{DXGL03} I. Djurdjevic, J. Xu, K. Abdel-Ghaffar, and S. Lin,
``A class of low-density parity-check codes constructed based on Reed-Solomon codes with two information symbols,''
\emph{IEEE Commun. Lett.}, vol. 7, no. 7, pp. 317--319, Jul. 2003.

\bibitem{Chen_Lan_Djurdjevic_Lin_04}
L. Chen, L. Lan, I. Djurdjevic, and S. Lin, ``An algebraic method for construction quasi-cyclic LDPC codes,''
Proc. \textit{Int. Symp. Inform. Theory and Its Applications}, Parma, Italy, Oct. 10--13, 2004, pp. 535--539.

\bibitem{TLG05} H. Tang, J. Xu, S. Lin, and K. A. S. Abdel-Ghaffar, ``Codes on finite geometries,'' 
\emph{IEEE Trans. Inform. Theory,} vol. 51, no. 2, 
pp. 572--596, Feb. 2005.


\bibitem{TLZLG06} Y. Y. Tai, L. Lan, L. Zheng, S. Lin and K. Abdel-Ghaffar, 
`` Algebraic construction of quasi-cyclic LDPC codes for the AWGN and erasure channels,'' 
\emph{IEEE Trans. Commun.}, vol 54, no. 7, pp. 1765--1774, Oct. 2006.


\bibitem{Xu-Chen-2007}
J. Xu, L. Chen, I. Djurdjevic, S. Lin, and K. Abdel-Ghaffar,
``Construction of regular and irregular LDPC codes: Geometry
decomposition and masking,''
 {\it IEEE Trans. Inform. Theory}, vol. 53, no. 1, pp. 121--134, Jan. 2007.


\bibitem{K07} N. Kamiya, ``High-rate quasi-cyclic low-density parity-check codes derived from finite affine planes,'' \emph{IEEE
          Trans. Inform. Theory}, vol. 53, no. 4, pp. 1444--1459, Apr. 2007.

 
\bibitem{LZTCLG07} L. Lan, L. Zeng, Y. Y. Tai, L. Chen, S. Lin, and K. Abdel-Ghaffar, ``Construction of
  quasi-cyclic LDPC codes for AWGN and binary erasure channels: A finite field approach,''
  \emph{IEEE Trans. Inform. Theory}, vol. 53, no. 7, pp. 2429--2458, Jul. 2007.



\bibitem{SZLG09} S. Song, B. Zhou, S. Lin, and K. Abdel-Ghaffar, 
``A unified approach to the construction of binary and nonbinary quasi-cyclic LDPC codes based on finite fields,'' 
\emph{IEEE Trans. Commun.}, vol. 57, no. 1, pp. 84--93, Jan. 2009.

\bibitem{JingyuKang2010}
J. Y. Kang, Q. Huang, L. Zhang, B. Zhou, and S. Lin, ``Quasi-Cyclic
LDPC Codes: An Algebraic Construction,'' {\it IEEE
Trans. Commun.}, vol. 58, no. 5, pp. 1383--1396, May 2010.


\bibitem{ZHLG10}  L. Zhang, Q. Huang, S. Lin, and  K. Abdel-Ghaffar, 
``Quasi-cyclic LDPC codes: An algebraic construction, rank analysis, and codes on Latin  squares,'' 
\emph{IEEE Trans. Commun.}, vol. 58, no. 11, pp. 3126--3139, Nov. 2010.



\bibitem{ZhangLinGhaffarDingZhou2011} L. Zhang, S. Lin, K. A. Ghaffar, Z. Ding, and B. Zhou, ``Quasi-cyclic LDPC codes on cyclic subgroups of finite fields,'' \emph{IEEE Trans. Commun.}, IEEE Trans. Commun., vol. 59, no. 9, pp. 2330-2336, Sep. 2011.

\bibitem{Huang_Diao_Lin_A-G}
Q. Huang, Q. Diao, S. Lin, and K. Abdel-Ghaffar, 
``Cyclic and quasi-cyclic LDPC codes on constrained parity-check matrices and their trapping sets,''
\emph{IEEE Trans. Inform. Theory}, to appear. 



\bibitem{HLG_12}
 Q. Diao, Q. Huang, S. Lin, and K. Abdel-Ghaffar, 
``A matrix theoretic approach for analyzing quasi-cyclic low-density parity-check codes,''
\emph{IEEE Trans. Inform. Theory}, to appear. 

\bibitem{HLG_10}
 Q. Diao, Q. Huang, S. Lin, and K. Abdel-Ghaffar, 
``A transform approach for computing the ranks of parity-check matrices of quasi-cyclic LDPC codes,'' \emph{Proc. 2011 IEEE Int. Symp. Inform. Theory}, SaintPetersburg, Russia, pp. 366-379, July 31-Aug. 5, 2011.


\bibitem{B83}  R. E. Blahut,  \emph{Theory and Practice of Error Control Codes}. Reading, MA: Addison-Wesley, 1983.


\bibitem{R06}  R. M. Roth,  \emph{Introduction to Coding Theory}. Cambridge, UK: Cambridge University Press, 2006.


 \bibitem{LN94} 
R. Lidl and H. Niederreiter, \emph{Introduction to Finite Fields and
their Applications}. revised ed. Cambridge, UK: Cambridge University Press, 1994.
         
         
\bibitem{B84} E. R. Berlekamp,  \emph{Algebraic Coding Theory}. New-York, NY: McGraw-Hill, 1964. 
(Rev. ed. Laguna Hills, CA: Aegean Park Press, 1984.
        
        
        
\bibitem{ZengLanTaiSongLinA-G_08} 
L. Zeng, L. Lan, Y. Y. Tai, S. Song, S. Lin, and K. Abdel-Ghaffar,
``Constructions of nonbinary quasi-cyclic LDPC codes: A finite field approach,''
{\it IEEE Trans. Commun.},
vol. 56, pp. 545--554, April 2008.



\bibitem{Million2007} 
Robert A. Beezer, ``A first course in linear algebra.'' \emph{Robert A. Beezer}, 1.08 edition, 2007.
         
\bibitem{HW1979}
G. H. Hardy and E. M. Wright,
\emph{An Introduction to the Theory of Numbers}, 5th Ed. Oxford, UK: Oxford University Press, 1979.         
         
         
\end{thebibliography}
\end{document}